%% file: main.tex
\documentclass[letterpaper,11pt]{article}

\usepackage{bm}
\usepackage{url}
\usepackage{cite}
\usepackage{array}
\usepackage{color}
\usepackage{amsthm}
\usepackage{dsfont}
\usepackage{amsmath}
\usepackage{amssymb}
\usepackage{authblk}
\usepackage{caption}
\usepackage{amsfonts}
\usepackage{booktabs}
\usepackage{graphicx}
\usepackage{mathrsfs}
\usepackage{multicol}
\usepackage{multirow}
\usepackage{cellspace}
\usepackage{enumerate}
\usepackage{enumitem}
\usepackage{makecell}

\usepackage{hyperref}
\hypersetup{
  linktoc      = all,
  colorlinks   = true,
  urlcolor     = blue,
  linkcolor    = blue,
  citecolor    = red
}
\usepackage[T1]{fontenc}
\usepackage[dvipsnames]{xcolor}
\usepackage[margin=1in]{geometry}
\usepackage[ruled,vlined,linesnumbered]{algorithm2e}
\usepackage[caption=false,font=normalsize,labelfont=sf,textfont=sf]{subfig}

\setlist[itemize]{leftmargin=*}
\setlist[enumerate]{leftmargin=*}

\title{Revisiting Local Computation of PageRank: Simple and Optimal\footnote{This text is the full version of a paper accepted by the 56th Annual ACM Symposium on Theory of Computing (STOC 2024). This work was partially done at Gaoling School of Artificial Intelligence, Beijing Key Laboratory of Big Data Management and Analysis Methods, MOE Key Lab of Data Engineering and Knowledge Engineering, and Pazhou Laboratory (Huangpu), Guangzhou, Guangdong 510555, China.}}

\author{Hanzhi Wang}
\author{Zhewei Wei\footnote{Zhewei Wei is the corresponding author.}}
\author{Ji-Rong Wen}
\author{Mingji Yang}

\affil{Renmin University of China \authorcr
  \{hanzhi\_wang, zhewei, jrwen, kyleyoung\}@ruc.edu.cn}

\date{}

\input{symbol-def.tex}

\begin{document}

\maketitle

\input{abstract}

\input{introduction}
\input{preliminaries}
\input{related_work}
\input{contribution_upper_bound}
\input{SNPR_upper_bound}
\input{contribution_lower_bound}
\input{SNPR_lower_bound}
\input{acknowledgments}

\appendix

\input{appendix}

\bibliographystyle{alphaurl}
\bibliography{paper}

\end{document}

%% file: symbol-def.tex
\newtheorem{theorem}{Theorem}[section]
\newtheorem{lemma}[theorem]{Lemma}

\def\header{\vspace{0.8mm} \noindent}

\DeclareMathOperator{\Var}{Var}

\DeclareMathOperator{\polylog}{polylog}

\newcommand{\indicator}[1]{\mathds{1}\{#1\}}

\def\tO{\widetilde{O}}

\def\Nin{\mathcal{N}_{\mathrm{in}}}
\def\Nout{\mathcal{N}_{\mathrm{out}}}
\def\din{d_{\mathrm{in}}}
\def\dout{d_{\mathrm{out}}}
\def\Deltain{\Delta_{\mathrm{in}}}
\def\Deltaout{\Delta_{\mathrm{out}}}

\def\rmax{\epsilon}

\newcommand{\epib}[1]{p(#1)}
\newcommand{\rb}[1]{r(#1)}

\def\epi{\hat{\pi}}
\def\pf{p_f}

\def\fpush{\textup{\texttt{ApproximatePageRank}}\xspace}
\def\bpush{\textup{\texttt{ApproxContributions}}\xspace}
\def\FASTPPR{\textup{\texttt{FAST-PPR}}\xspace}
\def\BiPPR{\textup{\texttt{BiPPR}}\xspace}
\def\unBiPPR{\textup{\texttt{Undirected-BiPPR}}\xspace}
\def\RBS{\textup{\texttt{RBS}}\xspace}
\def\SetPush{\textup{\texttt{SetPush}}\xspace}

\def\indeg{\textup{\textsc{indeg}}\xspace}
\def\parent{\textup{\textsc{parent}}\xspace}
\def\outdeg{\textup{\textsc{outdeg}}\xspace}
\def\child{\textup{\textsc{child}}\xspace}
\def\jump{\textup{\textsc{jump}}\xspace}
\def\neigh{\textup{\textsc{neigh}}\xspace}

\def\SampleNode{\textup{\texttt{SampleNode}}\xspace}

\def\smallexpo{\gamma}

\def\randint{\mathrm{randint}}

\def\SP{\mathrm{SP}}
\def\RP{\mathrm{RP}}

%% file: abstract.tex
\begin{abstract}

We revisit the classic local graph exploration algorithm \bpush proposed by Andersen, Borgs, Chayes, Hopcroft, Mirrokni, and Teng (WAW '07, Internet Math. '08) for computing an $\rmax$-approximation of the PageRank contribution vector for a target node $t$ on a graph with $n$ nodes and $m$ edges.
We give a worst-case complexity bound of \bpush as $O\left(n\pi(t)/\rmax\cdot\min\big(\Deltain,\Deltaout,\sqrt{m}\big)\right)$, where $\pi(t)$ is the PageRank score of $t$, and $\Deltain$ and $\Deltaout$ are the maximum in-degree and out-degree of the graph, resp.
We also give a lower bound of $\Omega\left(\min\big(\Deltain/\delta,\Deltaout/\delta,\sqrt{m}/\delta,m\big)\right)$ for detecting the $\delta$-contributing set of $t$, showing that the simple \bpush algorithm is already optimal.

As \bpush has become a cornerstone for computing random-walk probabilities, our results and techniques can be applied to derive better bounds for various relevant problems.
In particular, we investigate the computational complexity of locally estimating a node's PageRank centrality.
We improve the best-known upper bound of $\tO\left(n^{2/3}\cdot\min\left(\Deltaout^{1/3},m^{1/6}\right)\right)$ given by Bressan, Peserico, and Pretto (SICOMP '23) to $O\left(n^{1/2}\cdot\min\left(\Deltain^{1/2},\Deltaout^{1/2},m^{1/4}\right)\right)$ by simply combining \bpush with the Monte Carlo simulation method.
We also improve their lower bound of $\Omega\left(\min\left(n^{1/2}\Deltaout^{1/2},n^{1/3}m^{1/3}\right)\right)$ to $\Omega\left(n^{1/2}\cdot\min\left(\Deltain^{1/2},\Deltaout^{1/2},m^{1/4}\right)\right)$ if $\min(\Deltain,\Deltaout)=\Omega\left(n^{1/3}\right)$, and to $\Omega\left(n^{1/2-\smallexpo}\big(\min(\Deltain,\Deltaout)\big)^{1/2+\smallexpo}\right)$ if $\min(\Deltain,\Deltaout)=o\left(n^{1/3}\right)$, where $\smallexpo>0$ is an arbitrarily small constant.
Our matching upper and lower bounds resolve the open problem of whether one can tighten the bounds given by Bressan, Peserico, and Pretto (FOCS '18, SICOMP '23).
Remarkably, the techniques and analyses for proving all our results are surprisingly simple.

\end{abstract}

%% file: introduction.tex
\section{Introduction} \label{sec:introduction}

\textit{Random walks} on graphs are an algorithmic building block in modern network analysis.
The computation of random-walk probabilities serves as a versatile tool that has found multifaceted applications in webpage ranking~\cite{brin1998anatomy}, personalized search~\cite{haveliwala2003topic,fogaras2005towards}, graph partitioning~\cite{spielman2004nearly,andersen2007pagerank,fountoulakis2019variational}, graph sparsification~\cite{spielman2004nearly,spielman2011graph}, property testing~\cite{kale2011expansion,czumaj2015testing,chiplunkar2018testing}, and graph neural networks~\cite{klicpera2018predict,bojchevski2020scaling,wang2021approximate}, among many others.
\textit{Local graph algorithms}, which can estimate random-walk probabilities by only exploring a vanishing fraction of the graph, have been extensively studied.
Notably, the celebrated \textit{PageRank}~\cite{brin1998anatomy} centrality stands out as one of the most profound types of random-walk probabilities.
Roughly speaking, the PageRank centrality $\pi(t)$ of a node $t$ equals the probability that a random walk whose length follows a geometric distribution terminates at $t$.

We revisit the local graph exploration algorithm \bpush~\cite{andersen2007local,andersen2008local} (a.k.a. Backward Push, Backward Search, Reverse Local Push, etc.) originally proposed for computing \textit{PageRank contributions}, that is, for computing each node's contribution to a target node $t$'s PageRank centrality score $\pi(t)$.
The algorithm explores a subgraph near $t$ backward by performing a sequence of probability-pushing operations.
It takes as input a parameter $\rmax$ and computes an $\rmax$-approximation (in terms of maximum absolute error) of the PageRank contribution vector for $t$.
Although \bpush is a significant cornerstone algorithm, its computational complexity has not been fully understood despite years of studies.
Specifically, no easy-to-apply worst-case complexity bound of \bpush is known, which greatly limits its applications.
Meanwhile, whether one can derive an upper bound of $O\big(n\pi(t)/\rmax\big)$ for computing an $\rmax$-approximation of the PageRank contribution vector for $t$ remains an open problem.
Here, $n$ denotes the number of nodes in the graph, and the output size of the problem is $\Theta\big(n\pi(t)/\rmax\big)$ in the worst case.

We also investigate the classic problem of estimating a single target node $t$'s PageRank centrality $\pi(t)$.
This problem has been studied extensively for over a decade~\cite{chen2004local,fogaras2005towards,gleich2007approximating,andersen2008local,bar2008local,bressan2011local,borgs2012sublinear,borgs2014multiscale,lofgren2014fast,lofgren2015bidirectional,lofgren2016personalized,bressan2018sublinear,bressan2023sublinear}.
Before \cite{bressan2018sublinear}, the representative methods include Monte Carlo simulation, backward local exploration, and a combination of them, but none of them yield \textit{fully sublinear} worst-case upper bounds (i.e., bounds of $o(n+m)$, where $m$ denotes the number of edges in the graph) for any choice of $t$ on general directed graphs.
\cite{bressan2018sublinear,bressan2023sublinear} give the first fully sublinear upper bounds by designing a complicated backward exploration process, and they also give lower bounds for this problem.
However, there still exists a polynomial gap between their upper and lower bounds.

\subsection{Problem Formulations}

We consider an arbitrary directed graph $G=(V,E)$ with $n$ nodes and $m$ edges.
Let $\Deltain$ and $\Deltaout$ be the maximum in-degree and out-degree of $G$, resp.
We will parameterize the graph by $n$, $m$, $\Deltain$, and $\Deltaout$.
We assume that all nodes in $G$ have a nonzero out-degree for simplicity.
We measure the computational complexity under the standard RAM model, and we assume that the algorithms can have local access to the underlying graph, where they can learn at unit cost the in-degree/out-degree of a node or its $i$-th incoming/outgoing edge.
We also allow the algorithms to obtain a uniformly randomly selected node from $V$ through the $\jump()$ query at unit cost.
This setting is referred to as the \textit{arc-centric graph-access model} (see Section~\ref{sec:preliminaries_model} for details).

The definition of PageRank~\cite{brin1998anatomy} and PageRank contributions~\cite{andersen2008local} is based on the notion of \textit{$\alpha$-discounted random walk}, which is a variation of the standard random walk process that terminates w.p. $\alpha$ at each step independently.
We regard the \textit{decay factor} $\alpha$ as a fixed constant.
The PageRank contribution of $v$ to $t$ equals the probability that an $\alpha$-discounted random walk from $v$ terminates at $t$, and the PageRank score of $t$, denoted by $\pi(t)$, equals the average of the PageRank contributions from all nodes in $V$ to $t$.

We consider two problems relevant to estimating PageRank contributions: (i) computing an $\rmax$-approximation of the PageRank contribution vector for $t$, where we aim to estimate all nodes' PageRank contributions to $t$ with absolute error at most $\rmax$, and (ii) detecting the \textit{$\delta$-contributing set} of $t$, where we aim to compute a node set containing all nodes whose PageRank contributions to $t$ are at least $\delta n\pi(t)$.
For single-node PageRank estimation, our goal is to compute a multiplicative $(1\pm c)$-approximation of $\pi(t)$ w.p. at least $(1-\pf)$.
Here, $\rmax$, $\delta$, $c$, and $\pf$ are parameters in $(0,1)$.

\subsection{Our Results}

We give matching upper and lower bounds for the worst-case computational complexity of estimating PageRank contributions and single-node PageRank.
Our first result is a tight upper bound for computing an $\rmax$-approximation of the contribution vector for node $t$.

\begin{theorem} \protect{\textup{[}\hyperlink{proof:cost_bp}{$\downarrow$}\textup{]}\footnote{These downwards arrows are hyperlinks to the proofs.}} \label{thm:cost_bp}
    The computational complexity of computing an $\rmax$-approximation of the contribution vector for a target node $t$ without using the $\jump()$ query is $O\left(n\pi(t)/\rmax\cdot\min\big(\Deltain,\Deltaout,\sqrt{m}\big)\right)$.
\end{theorem}

\noindent
As we shall see, this upper bound is derived by revisiting the complexity of \bpush.
This is the first easy-to-apply worst-case complexity bound of \bpush despite years of studies.
Previously, only a somewhat convoluted worst-case bound was known, leading to a preference for analyzing algorithms involving \bpush in terms of the average-case complexity as opposed to the worst-case one~\cite{lofgren2014fast,lofgren2016personalized,zhang2016approximate,wang2016hubppr,wang2018efficient,liao2022efficient} (see Section~\ref{sec:preliminaries_bpush} for detailed discussions).
In contrast, our bound succinctly captures the dependence of \bpush's cost on $t$ by $\pi(t)$.

One may note that this upper bound is still inferior to the desired upper bound of $O\big(n\pi(t)/\rmax\big)$.
However, we show that such an upper bound is impossible to achieve unless $\min\left(\Deltain,\Deltaout\right)=\Theta(1)$, and, surprisingly, the original \bpush algorithm is already optimal.
To this end, we study the problem of detecting the $\delta$-contributing set of $t$.
It is known that any algorithm that computes an $\rmax$-approximation of the PageRank contribution vector, including \bpush, can be applied to this problem.
Using Theorem~\ref{thm:cost_bp}, we derive the following upper bound for detecting the $\delta$-contributing set.

\begin{theorem} \protect{\textup{[}\hyperlink{proof:contribution_upper_bound}{$\downarrow$}\textup{]}} \label{thm:contribution_upper_bound}
    The computational complexity of computing a superset of the $\delta$-contributing set of a target node $t$ without using the $\jump()$ query is $O\left(\min\big(\Deltain/\delta,\Deltaout/\delta,\sqrt{m}/\delta,m\big)\right)$.
\end{theorem}

\noindent
At first glance, this theorem may be derived by simply applying Theorem~\ref{thm:cost_bp} with $\rmax=\Theta\big(\delta n\pi(t)\big)$.
Nonetheless, we emphasize that our result holds even if we only know $t$ and $\delta$, while the original paper~\cite{andersen2007local} assumes that $n\pi(t)$ is known when solving this problem.
On the other hand, we prove the following matching lower bound for detecting the $\delta$-contributing set.

\begin{theorem}[Informal] \label{thm:contribution_lower_bound_informal}
    $\Omega\left(\min\big(\Deltain/\delta,\Deltaout/\delta,\sqrt{m}/\delta,m\big)\right)$ queries are in general required to detect the $\delta$-contributing set of $t$ without using the $\jump()$ query.
\end{theorem}

\noindent
The formal version of this theorem is given as Theorem~\ref{thm:contribution_lower_bound} in Section~\ref{sec:contribution_lower_bound}.
This lower bound refutes the possibility of achieving the upper bound of $O(1/\delta)$ for detecting the $\delta$-contributing set and the corresponding $O\big(n\pi(t)/\rmax\big)$ bound for computing an $\rmax$-approximation of the PageRank contribution vector, unless $\min(\Deltain,\Deltaout)=\Theta(1)$.
This lower bound also establishes the optimality of the simple \bpush algorithm, which has remained unrevealed despite years of studies.

As \bpush has become a fundamental building block for computing random-walk probabilities, our results and techniques can potentially be applied to derive better bounds for various relevant algorithms and to prove their optimality.
In particular, we further study the problem of single-node PageRank estimation.
By simply combining \bpush with the standard \textit{Monte Carlo simulation} method, we obtain the following upper bound for this problem.

\begin{theorem} \protect{\textup{[}\hyperlink{proof:SNPR_upper_bound}{$\downarrow$}\textup{]}} \label{thm:SNPR_upper_bound}
Given a target node $t$, the expected computational complexity of computing a multiplicative $(1\pm c)$-approximation of $\pi(t)$ w.p. at least $(1-\pf)$ is $O\left(n^{1/2}\cdot\min\left(\Deltain^{1/2},\Deltaout^{1/2},m^{1/4}\right)\right)$\footnote{Throughout the paper, for readability we omit multiplicative factors depending only on the approximation parameters $c$ and $\pf$ in the $O()$ notation for the complexity of single-node PageRank estimation, following \cite{bressan2023sublinear}. These factors are mildly polynomial in $1/c$ and polylogarithmic in $1/\pf$ and can be found in the \hyperlink{proof:SNPR_upper_bound}{proof} of Theorem~\ref{thm:SNPR_upper_bound}.}.
\end{theorem}

\noindent
The idea of combining \bpush with Monte Carlo simulation has been adopted in the \BiPPR algorithm~\cite{lofgren2016personalized}.
However, \cite{lofgren2016personalized} only gives an average-case upper bound of $O\big(\sqrt{m}\big)$ over all $t\in V$, which cannot hold in the worst case as indicated by the known lower bounds.
In contrast, our upper bound is a worst-case one, and it improves over the currently best bound of $\tO\left(n^{2/3}\cdot\min\left(\Deltaout^{1/3},m^{1/6}\right)\right)$ given by \cite{bressan2023sublinear} for computational complexity.
Note that \cite{bressan2023sublinear} also gives a bound of $\tO\left(n^{1/2}\cdot\min\left(\Deltaout^{1/2},m^{1/4}\right)\right)$, but it only holds for query complexity (i.e., the number of queries made to the oracle, see Section~\ref{sec:preliminaries_model} for details), while our upper bound holds for both computational complexity and query complexity and is free from $\polylog(n)$ factors.
Our upper bound shows that one can always solve this problem in a \textit{sublinear} time of $O\left((n+m)^{3/4}\right)$, and also in $o(n)$ time as long as $m=o\left(n^2\right)$.

Our final contributions are matching lower bounds for single-node PageRank estimation.
We first study the easier case when the $\Deltain$ parameter is excluded, as considered by \cite{bressan2018sublinear,bressan2023sublinear}.

\begin{theorem}[Informal] \label{thm:SNPR_lower_bound_informal}
    $\Omega\left(n^{1/2}\cdot\min\left(\Deltaout^{1/2},m^{1/4}\right)\right)$ queries are in general required to estimate $\pi(t)$ within a multiplicative factor of $O(1)$ w.p. $\Omega(1)$.
\end{theorem}

\noindent
The formal version of this theorem is given as Theorem~\ref{thm:SNPR_lower_bound} in Section~\ref{sec:SNPR_lower_bound}.
This lower bound improves over the best-known lower bound of $\Omega\left(\min\left(n^{1/2}\Deltaout^{1/2},n^{1/3}m^{1/3}\right)\right)$ given by \cite{bressan2018sublinear,bressan2023sublinear} and matches our upper bound of $O\left(n^{1/2}\cdot\min\left(\Deltaout^{1/2},m^{1/4}\right)\right)$ if we disregard $\Deltain$.
In fact, in \cite{bressan2018sublinear,bressan2023sublinear}, the authors only prove their lower bound when $\Deltaout=\Theta\left(n^{-1/3}m^{2/3}\right)$, and thus, their lower bound is not proved to hold for an arbitrary combination of $\Deltaout$ and $m$, and in particular it does not cover the case when $\Deltaout=o\left(n^{1/3}\right)$.
In contrast, we give tighter lower bounds with more rigorous arguments, covering all possibilities of the parameters.
Note that this lower bound also shows that the previous bound of $\tO\left(n^{1/2}\cdot\min\left(\Deltaout^{1/2},m^{1/4}\right)\right)$ given by \cite{bressan2023sublinear} for query complexity is already tight up to polylogarithmic factors.
Next, the following theorem further incorporates $\Deltain$ into the lower bound.

\begin{theorem}[Informal] \label{thm:SNPR_lower_bound_Delta_in_informal}
    If $\min(\Deltain,\Deltaout)=\Omega\left(n^{1/3}\right)$, then $\Omega\left(n^{1/2}\cdot\min\left(\Deltain^{1/2},\Deltaout^{1/2},m^{1/4}\right)\right)$ queries are in general required to estimate $\pi(t)$ within a multiplicative factor of $O(1)$ w.p. $\Omega(1)$.

    If instead $\min(\Deltain,\Deltaout)=o\left(n^{1/3}\right)$, then $\Omega\left(n^{1/2-\smallexpo}\big(\min(\Deltain,\Deltaout)\big)^{1/2+\smallexpo}\right)$ queries are in general required to estimate $\pi(t)$ within a multiplicative factor of $O(1)$ w.p. $\Omega(1)$, where $\smallexpo>0$ is an arbitrarily small constant.
\end{theorem}

\noindent
The formal version of this theorem is given as Theorem~\ref{thm:SNPR_lower_bound_Delta_in} in Section~\ref{sec:SNPR_lower_bound}.
These lower bounds refute the existence of an upper bound that is polynomially smaller than $n^{1/2}\cdot\min\left(\Deltain^{1/2},\Deltaout^{1/2},m^{1/4}\right)$ and demonstrate that our upper bound of $O\left(n^{1/2}\cdot\min\left(\Deltain^{1/2},\Deltaout^{1/2},m^{1/4}\right)\right)$ is optimal up to subpolynomial factors.
As a special case, when $\Deltain=\Theta(1)$ and $\Deltaout=\Theta(1)$, our upper and lower bounds become $O\left(n^{1/2}\right)$ and $\Omega\left(n^{1/2-\smallexpo}\right)$ for any constant $\smallexpo>0$, resp.
Our results bridge the gap between the previously known upper and lower bounds for single-node PageRank estimation under the arc-centric graph-access model.

\subsection{Technical Overview}

Our techniques and analyses for deriving the results are surprisingly simple.
Our algorithms for proving the upper bounds are basically the original \bpush~\cite{andersen2008local} and \BiPPR~\cite{lofgren2016personalized} algorithms, save that we may use different parameter settings and apply the \textit{doubling} technique to invoke them several times.
The \bpush algorithm works by performing a sequence of \textit{pushback} operations on some nodes to propagate the probability mass backward along the incoming edges, and the \BiPPR algorithm estimates single-node PageRank by combining \bpush with the standard Monte Carlo simulation method, which simulates a number of $\alpha$-discounted random walks on the graph.

By carefully reanalyzing the complexity of \bpush, we find a simple and short proof for the tight upper bound of its worst-case complexity, which has remained undiscovered for more than a decade.
We then apply \bpush to detect the contributing set through an \textit{adaptive} setting of its parameters using the doubling technique, where we exploit the property of the pushback operations to determine the stopping rule.
For single-node PageRank estimation, \BiPPR uses the key \textit{invariant} property of \bpush to construct a low-variance \textit{bidirectional estimator} of $\pi(t)$ as a linear combination of the Monte Carlo estimators.
\BiPPR simply generates several samples of the bidirectional estimator by simulating random walks and computes their average as the final estimate.
We revisit \BiPPR by bounding the variance of the bidirectional estimator to reanalyze its concentration behavior and re-balancing the cost of \bpush and Monte Carlo simulation, leading to a tight upper bound for \BiPPR's computational complexity.

Remarkably, our algorithm and analyses for single-node PageRank estimation are both dramatically simpler than those given by \cite{bressan2018sublinear,bressan2023sublinear}, where the authors devise a complex local exploration method to construct a \textit{perfect weighted estimator} with desirable concentration behavior.
Our analyses show that the simple \bpush algorithm provides an elegant way to construct a ``more perfect'' bidirectional estimator using minimal cost.
Therefore, the complex local exploration method in \cite{bressan2018sublinear,bressan2023sublinear} is indeed unnecessary.
We give more detailed discussions on this at the end of Section~\ref{sec:SNPR_upper_bound}.
This comparison further showcases the power of \bpush: it is more than an optimal ``black-box'' algorithm for estimating PageRank contributions; instead, it also provides excellent properties that enable it to be seamlessly combined with other techniques.

As for lower bounds, we construct carefully designed hard-instance graphs and argue the least number of queries the algorithm needs to invoke to solve the problem on them.
For detecting the contributing set, we construct hard-to-detect nodes that lie in the target contributing set; for single-node PageRank estimation, we construct a hard-to-detect portion of the graph that largely determines the magnitude of $\pi(t)$.
The underlying difficulty comes from detecting a special unseen parent of a node among a number of its parents under the arc-centric graph-access model, even if the remaining parents of the node are already explored by the algorithm through other nodes.
Thus, the algorithm must inspect a large number of edges on the hard-instance graphs.
Our proofs cover all possible relationships among the considered parameters, including $\Deltain$, $\Deltaout$, and $m$.
In particular, when $\Deltain$ is small, we embed a \textit{multi-level structure} with constant maximum in-degree into the hard-instance graphs to replace nodes with large in-degrees.

\header{\textbf{Extensions.}}
Given the numerous extensions and applications of \bpush over the years~\cite{banerjee2015fast,wang2016hubppr,zhang2016approximate,guo2017parallel,wang2018efficient,wei2018topppr,wei2019prsim,yin2019scalable,chen2020scalable,wang2021approximate,liao2022efficient,zheng2022instant,mo2023single}, our results and techniques can potentially be applied to reanalyze various relevant algorithms for computing random-walk probabilities, giving tighter bounds for other problems beyond PageRank computation.
A notable extension of \bpush is \cite{banerjee2015fast}, which modifies the algorithm slightly to estimate \textit{Markov-chain multi-step transition probabilities} and \textit{graph diffusions}, including the \textit{Heat Kernel} centrality~\cite{chung2007heat}.
Their modification partitions the probability-pushing process of \bpush into different levels but remains a similar analytical framework. By analyzing the complexity of running \bpush level by level, one can potentially apply our results and techniques to derive a tighter upper bound for estimating multi-step transition probabilities from a source node to a target node. We also remark that the techniques and arguments for establishing our lower bounds can be extended to the problem of computing multi-step transition probabilities under the arc-centric graph-access model.

\subsection{Paper Organization}

The remainder of this paper is organized as follows.
Section~\ref{sec:preliminaries} provides preliminaries and notations, including descriptions of the considered models and known results for the \bpush algorithm.
Section~\ref{sec:related_work} discusses some related work.
We present the proofs of our main results in Sections~\ref{sec:contribution_upper_bound} to \ref{sec:SNPR_lower_bound}, where we prove all the upper bounds before delving into the lower bounds to ensure a clear and logical progression.
In particular, Sections~\ref{sec:contribution_upper_bound} and \ref{sec:SNPR_upper_bound} prove the upper bounds for estimating PageRank contributions and single-node PageRank, resp., and Sections~\ref{sec:contribution_lower_bound} and \ref{sec:SNPR_lower_bound} prove the lower bounds for estimating PageRank contributions and single-node PageRank, resp.
Furthermore, we offer a table of notations along with pseudocodes in Appendix~\ref{sec:table_notations}, and give the deferred proofs in Appendix~\ref{sec:deferred_proofs}.

%% file: preliminaries.tex
\section{Preliminaries and Notations} \label{sec:preliminaries}

We denote the underlying graph by $G=(V,E)$ with $n=|V|$ and $m=|E|$.
If $(u,v)\in E$, we say that $u$ is a parent of $v$ and $v$ is a child of $u$.
We use $\Nin(v)$ and $\Nout(v)$ to represent the set of parents and children of a node $v$, resp.
We denote by $\din(v)$ and $\dout(v)$ the in-degree and out-degree of $v$, resp.
We define $\Deltain=\max_{v\in V}\din(v)$ and $\Deltaout=\max_{v\in V}\dout(v)$.
We consider random walks that move to a uniformly random child of the current node at each step.

\subsection{PageRank and PPR} \label{sec:preliminaries_PageRank}

It is well-known that the \textit{PageRank score} of $v$, denoted by $\pi(v)$, equals the probability that an \textit{$\alpha$-discounted random walk} from a uniformly random source node in $V$ terminates at $v$~\cite{fogaras2005towards,avrachenkov2007monte}.
Here, $\alpha$ is a constant \textit{decay factor} in $(0,1)$, and an $\alpha$-discounted random walk is defined as a random walk whose length is a random variable that takes on value $\ell$ w.p. $\alpha(1-\alpha)^{\ell}$ for each $\ell\ge0$.
By definition, the following recursive equality holds for each $v\in V$:
\begin{align} \label{eqn:iterative_pagerank}
    \pi(v)=\sum_{u\in \Nin(v)}\frac{(1-\alpha)\cdot \pi(u)}{\dout(u)}+\frac{\alpha}{n}.
\end{align}
The \textit{Personalized PageRank (PPR) score}~\cite{brin1998anatomy} of a target node $v$ w.r.t. a source node $u$, denoted by $\pi(u,v)$, equals the probability that an $\alpha$-discounted random walk from $u$ terminates at $v$.
It is straightforward that $\pi(v)=1/n\cdot\sum_{u\in V}\pi(u,v)$ for any $v$.

\subsection{PageRank Contributions and the Contributing Set} \label{sec:preliminaries_contribution}

The ``\textit{PageRank contribution} of node $v$ to target node $t$'' is a synonym of the ``PPR score of $t$ w.r.t. $v$,'' i.e., the value of $\pi(v,t)$.
The \textit{PageRank contribution vector} of $t$ is the $n$-dimensional vector recording each node $v$'s contribution to $t$.
The \textit{$\rmax$-approximation} of the PageRank contribution vector for $t$ is a (sparse) vector whose difference from the contribution vector is at most $\rmax$ at each node.
The \textit{$\delta$-contributing set} of $t$ is defined as the node set $\big\{v\in V\mid\pi(v,t)\ge\delta n\pi(t)\big\}$.
These terminologies come from the seminal paper~\cite{andersen2008local} (note that their notation $\text{pr}(t)$ corresponds to our $n\pi(t))$.

\subsection{The Graph-Access Model} \label{sec:preliminaries_model}

For our results, we focus on the \textit{computational complexity} under the standard RAM model.
To clarify what operations for interacting with the graph are available, we enable the local algorithms to access the underlying graph $G$ through a graph oracle.
We consider the standard \textit{arc-centric graph-access model}~\cite{goldreich1998property,goldreich2002property}, where the graph oracle supports the following queries in unit time: $\indeg(v)$, which returns $\din(v)$; $\outdeg(v)$, which returns $\dout(v)$; $\parent(v,i)$, which returns the $i$-th parent of $v$; $\child(v,i)$, which returns the $i$-th child of $v$; $\jump()$, which returns a node in $V$ chosen uniformly at random.
We call $\jump()$ the \textit{global query} and others the \textit{local queries}.
This model is also used by the relevant works of \cite{lofgren2014fast,lofgren2015bidirectional,lofgren2016personalized,bressan2018sublinear,bressan2023sublinear}.
When considering local exploration algorithms for computing PageRank contributions, we only allow the local queries, as explicitly specified in the given theorems.

\header{\textbf{Query Complexity.}}
While this paper mainly concentrates on computational complexity, we explain \textit{query complexity} here for clarity.
In the context of local graph algorithms, the query complexity is defined as the number of queries the algorithm invokes to the graph oracle.
Thus, query complexity is a lower bound to computational complexity, and our upper and lower bounds straightforwardly apply to both computational and query complexities.

\header{\textbf{Node-Centric Graph-Access Model.}}
In our discussions, we also mention the \textit{node-centric graph-access model} as considered in \cite{bar2008local,bressan2013power,bressan2018sublinear,bressan2023sublinear}.
This model allows two queries: the $\jump()$ query and the powerful $\neigh(v)$ query that returns all parents and children of $v$ in one shot.
Clearly, $n$ queries are always sufficient to explore the whole graph under this model, and the query complexity under this model is never larger than that under the arc-centric model.

\subsection{Known Results for \bpush} \label{sec:preliminaries_bpush}

This subsection introduces some previously known results for \bpush~\cite{andersen2008local} and points out their major drawback.
Overall, \bpush approximates PageRank contributions from below by performing a sequence of \textit{pushback} operations on some nodes, pushing probability mass backward along their incoming edges.
In fact, \bpush is similar to the celebrated forward exploration algorithm \fpush~\cite{andersen2006local,andersen2007pagerank}, which also works by performing a sequence of probability-pushing operations but in the forward direction.

We give the pseudocode of \bpush as Algorithm~\ref{alg:BP} on the next page.
The algorithm takes as input the target node $t$ and a threshold $\rmax\in(0,1]$.
It maintains two vectors $\epib{}$ and $\rb{}$ for nodes in $V$ (implemented by dictionaries), which we respectively call the \textit{reserve} vector and the \textit{residue} vector for ease of reference.
After initializing the vectors to all zeros except that $\rb{t}=1$, the algorithm repeatedly performs pushback operations (to be explained shortly) on nodes $v$ with $\rb{v}>\rmax$.
When $\rb{v}\le\rmax$ holds for each $v$, \bpush terminates and returns $\epib{}$ and $\rb{}$, where the reserve vector $\epib{}$ contains underestimates for the PageRank contributions of $t$ with absolute error at most $\rmax$ at each node.
In a pushback operation on $v$, the algorithm transfers $\alpha$ fraction of the residue $\rb{v}$ to its reserve $\epib{v}$, propagates the remaining residue to each parent of $v$, and resets $\rb{v}$ to $0$.
When the residue of $v$ is propagated to its parent $u$, $\rb{u}$ is increased by $(1-\alpha)\rb{v}/\dout(u)$.
We emphasize that \bpush does not specify any particular order of the pushback operations.

\begin{algorithm}[ht]
    \DontPrintSemicolon
    \caption{$\bpush(t,\alpha,\rmax)$~\cite{andersen2008local}} \label{alg:BP}
    \KwIn{target node $t\in V$, decay factor $\alpha$, threshold $\rmax$}
    \KwOut{dictionary $\epib{}$ for reserves and dictionary $\rb{}$ for residues}
    $\epib{},\rb{}\gets$ empty dictionaries with default value $0$ \;
    $\rb{t}\gets1$ \;
    \While{\textup{there exists a node} $v$ \textup{with} $\rb{v}>\rmax$}
    {
        pick an arbitrary node $v$ with $\rb{v}>\rmax$ \;
        $r\gets \rb{v}$ $\quad\quad\quad$ \textcolor{gray}{// temporary variable for storing $\rb{v}$} \; \label{line:push_begin}
        $\epib{v}\gets\epib{v}+\alpha r$ \;
        $\rb{v}\gets0$ \;
        \For{$i$ \textbf{\textup{from}} $1$ \textbf{\textup{to}} $\textsc{indeg}(v)$}
        {
            $u\gets\textsc{parent}(v,i)$ \;
            $\rb{u}\gets\rb{u}+(1-\alpha)r/\textsc{outdeg}(u)$ \; \label{line:push_end}
        }
    }
    \Return $\epib{}$ and $\rb{}$ \;
\end{algorithm}

The key property of \bpush is that the following \textit{invariant} is maintained by the pushback operations.

\begin{lemma}[\protect{Invariant~\cite[Lemma 1]{andersen2008local}}] \label{lem:invariant_backward}
    For the target node $t$, the pushback operations in the \bpush algorithm maintain the following invariant for each $s\in V$:
    \begin{align} \label{eqn:invariant_backward}
        \pi(s,t)=\epib{s}+\sum_{v\in V}\pi(s,v)\cdot\rb{v}.
    \end{align}
\end{lemma}

\noindent
The following theorem summarizes some previously known properties of \bpush.

\begin{theorem}[\protect{\cite[Theorem 1]{andersen2008local}}] \label{thm:properties_bp}
    $\bpush(t,\alpha,\rmax)$ computes nonnegative reserves $\epib{}$ such that $\pi(v,t)-\rmax\le\epib{v}\le\pi(v,t)$ holds for any $v\in V$ by performing $O\big(n\pi(t)/\rmax\big)$ pushback operations.
\end{theorem}

\noindent
We remark that the number of pushback operations performed offers an upper bound on the number of nodes $v$ with nonzero $\epib{v}$.

Most importantly, the original paper~\cite{andersen2008local} does not give a detailed computational complexity bound of \bpush.
Instead, they only state that its complexity is bounded by the sum of the in-degrees of the sequence of nodes on which the pushback operations are performed.
Nevertheless, we have the following complexity bounds of \bpush revealed by subsequent works.

\begin{theorem}[\cite{lofgren2013personalized,wang2020personalized}] \label{thm:previous_cost_bp}
    The complexity of running \bpush for a target node $t$ with parameter $\rmax\in(0,1]$ is bounded by both
    \begin{align} \label{eqn:previous_cost_bp}
        O\left(\frac{1}{\rmax}\sum_{v\in V}\pi(v,t)\cdot\din(v)\right)\quad\text{and}\quad O\left(\frac{1}{\rmax}\sum_{v\in V}\pi(v,t)\cdot\dout(v)\right).
    \end{align}
\end{theorem}

\noindent
For completeness, we give a \hyperlink{proof:previous_cost_bp}{proof} of this theorem in Appendix~\ref{sec:deferred_proofs}.

While the procedure of \bpush and its other properties are natural and elegant, its complexity bound is somewhat convoluted and not fully understood in the past years.
It is worth mentioning that there exist some folklore bounds for the complexity of \bpush in some special cases: its average complexity over all target nodes $t\in V$ is $O\big(m/(n\rmax)\big)$~\cite{lofgren2013personalized}, and its complexity on undirected graphs is $O\big(d(t)/\rmax\big)$, where $d(t)$ is the degree of $t$~\cite{wang2023estimating}.
However, the bounds in Theorem~\ref{thm:previous_cost_bp} are still the only known complexity bounds for an arbitrary target node $t$ on general directed graphs.
Notably, the summations $\sum_{v\in V}\pi(v,t)\cdot\din(v)$ and $\sum_{v\in V}\pi(v,t)\cdot\dout(v)$ depend on $t$ but cannot be expressed by simple parameters of $t$, and they can reach $\Theta(m)$ in the worst case, making these bounds difficult to apply to the worst-case analysis.

%% file: related_work.tex
\section{Related Work} \label{sec:related_work}

This section briefly reviews relevant works for locally estimating PageRank contributions and single-node PageRank.

\subsection{Estimating PageRank Contributions} \label{sec:related_work_contributions}

The problem of estimating PageRank contributions is introduced in \cite{andersen2007local,andersen2008local} and, along with its extensions, has found numerous applications over the years~\cite{andersen2008robust,banerjee2015fast,zhang2016approximate,wei2018topppr,chen2020scalable,wang2021approximate}.
As introduced in Section~\ref{sec:preliminaries_bpush}, the classic \bpush algorithm is proposed in the seminal paper~\cite{andersen2008local}, with its worst-case complexity bounds given by subsequent works of \cite{lofgren2013personalized,wang2020personalized}, and some folklore complexity bounds exist for the average case or undirected graphs~\cite{lofgren2013personalized,wang2023estimating}.
However, there still lacks a meaningful and easy-to-apply worst-case complexity bound of \bpush, which greatly limits the analysis for some subsequent algorithms involving it~\cite{lofgren2014fast,banerjee2015fast,lofgren2016personalized,wang2018efficient,liao2022efficient}.

Aside from \bpush, few algorithms have been proposed for estimating PageRank contributions due to the hardness of the problem.
Whether it is possible to detect the $\delta$-contributing set in $O(1/\delta)$ time has remained an open problem for more than a decade.
A recent work~\cite{wang2020personalized} gives a promising result to this problem: it proposes a randomized algorithm called \RBS, which detects the $\delta$-contributing set in $\tO(1/\delta)$ time by applying randomness to the pushback operations.
However, \RBS requires $\Theta(n+m)$ time and space for preprocessing to sort the adjacency list of each node in some particular order, which is not a local process.
In fact, our lower bound shows that this $\tO(1/\delta)$ bound is in general impossible to achieve without sufficient preprocessing.
Our results, on the other hand, show that the simple \bpush algorithm is already optimal.

\subsection{Estimating Single-Node PageRank} \label{sec:related_work_SNPR}

The problem of estimating single-node PageRank is initially introduced in \cite{chen2004local} and has attracted extensive studies since then~\cite{fogaras2005towards,gleich2007approximating,andersen2007local,andersen2008local,bar2008local,bressan2011local,borgs2012sublinear,borgs2014multiscale,lofgren2014fast,lofgren2015bidirectional,lofgren2016personalized,bressan2018sublinear,bressan2023sublinear}.
While these works cover a wide range of settings, we primarily concentrate on the worst-case computational complexity for this problem under the arc-centric graph-access model.

First, a line of work~\cite{fogaras2005towards,avrachenkov2007monte,borgs2012sublinear,borgs2014multiscale} uses Monte Carlo simulation to estimate $\pi(t)$ by generating $\alpha$-discounted random walks.
It is well-known that $\Theta\big(1/\pi(t)\big)$ samples are both necessary and sufficient for Monte Carlo simulation.
However, this cost can scale linearly in $(n+m)$ if $\pi(t)=\Theta(1/n)$ and $m=\Theta(n)$.
On the other hand, the local exploration algorithm \bpush~\cite{andersen2008local} can also be applied to estimate $\pi(t)$, but it is shown afterward that $o(n+m)$ time is impossible to achieve without using the global $\jump()$ query, even if the algorithm is randomized and is allowed to fail~\cite{bar2008local,bressan2013power}.

A promising approach to estimate $\pi(t)$ in sublinear time is to combine the local exploration method and Monte Carlo simulation.
Using this idea, \FASTPPR~\cite{lofgren2014fast} and \BiPPR~\cite{lofgren2016personalized} achieve an average complexity of $\tO\big(\sqrt{m}\big)$ over all $t\in V$, but with no worst-case guarantees for an arbitrary choice of $t$.
Another work called \unBiPPR~\cite{lofgren2015bidirectional} achieves a worst-case complexity of $O\left(\sqrt{nd(t)}\right)$ on undirected graphs, where $d(t)$ denotes the degree of $t$.
However, it cannot be applied to general directed graphs since its analysis heavily relies on a symmetry property of PPR on undirected graphs.
Additionally, a recent work called \SetPush~\cite{wang2023estimating} achieves a better complexity bound of $\tO\left(\min\big(d(t),\sqrt{m}\big)\right)$ for estimating $\pi(t)$ on undirected graphs.
Similarly, \SetPush cannot be applied to general directed graphs; in fact, it does not use the global $\jump()$ query and is thus impossible to break the $\Theta(n+m)$ barrier on directed graphs.

The first fully sublinear upper bound for the computational complexity of this problem is the $\tO\left(\min\left(n^{3/4}\Deltaout^{1/4},n^{5/7}m^{1/7}\right)\right)$ bound given by \cite{bressan2018sublinear}.
This upper bound is later improved to $\tO\left(n^{2/3}\cdot\min\left(\Deltaout^{1/3},m^{1/6}\right)\right)$ in \cite{bressan2023sublinear} by tightening the analysis.
We emphasize that although \cite{bressan2023sublinear} also gives an upper bound of $\tO\left(n^{1/2}\cdot\min\left(\Deltaout^{1/2},m^{1/4}\right)\right)$, it only holds for the query complexity instead of the computational complexity.
Their upper bounds are derived also using a combination of local exploration and Monte Carlo simulation.
They use a novel and complicated local exploration method to construct a \textit{perfect weighted estimator}, which is a linear combination of the Monte Carlo estimators with controllable coefficients.
They also use a technique called \textit{blacklisting} to avoid nodes with large in-degrees during the local exploration process.
Furthermore, their algorithm needs to compute approximate coefficients to construct \textit{approximate estimators}, leading to extra computational cost and analytical difficulty.
On the other hand, the best-known lower bound for the setting in question is $\Omega\left(\min\left(n^{1/2}\Deltaout^{1/2},n^{1/3}m^{1/3}\right)\right)$ given by \cite{bressan2018sublinear,bressan2023sublinear}.
The question of whether it is possible to tighten these upper and lower bounds was left as an open problem in \cite{bressan2018sublinear,bressan2023sublinear} and is resolved by our results.

We remark that many other upper and lower bounds under different settings exist for this problem.
We refer interested readers to \cite{bressan2023sublinear} for a more comprehensive summary of them.

%% file: contribution_upper_bound.tex
\section{Revisiting the Complexity of \bpush} \label{sec:contribution_upper_bound}

This section proves Theorem~\ref{thm:cost_bp} and sketches Theorem~\ref{thm:contribution_upper_bound} (the \hyperlink{proof:contribution_upper_bound}{proof} of Theorem~\ref{thm:contribution_upper_bound} is given in Appendix~\ref{sec:deferred_proofs}).
Our first result, Theorem~\ref{thm:cost_bp}, gives a new complexity bound of \bpush parameterized by $\pi(t)$, $\rmax$, and the parameters of $G$.

\begin{proof}[Proof of Theorem~\ref{thm:cost_bp}] \hypertarget{proof:cost_bp}
By Theorem~\ref{thm:properties_bp}, it suffices to prove that the computational complexity of $\bpush(t,\alpha,\rmax)$ is $O\left(n\pi(t)/\rmax\cdot\min\big(\Deltain,\Deltaout,\sqrt{m}\big)\right)$.
We first prove the upper bound of $O\big(n\pi(t)/\rmax\cdot\sqrt{m}\big)$ by manipulating the bound of $O\left(\sum_{v\in V}\pi(v,t)\cdot\din(v)/\rmax\right)$ in Theorem~\ref{thm:previous_cost_bp}.
For the summation in the bound, we have
\begin{align*}
    \sum_{v\in V}\pi(v,t)\cdot\din(v)&=\sum_{v\in V}\pi(v,t)\sum_{u\in\Nin(v)}\sqrt{\dout(u)}\cdot\frac{1}{\sqrt{\dout(u)}} \\
    &=\sum_{v\in V}\sum_{u\in\Nin(v)}\sqrt{\pi(v,t)\cdot\dout(u)}\cdot\sqrt{\frac{\pi(v,t)}{\dout(u)}}.
\end{align*}
Applying Cauchy-Schwarz inequality and reordering the terms, we obtain
\begin{align*}
    \sum_{v\in V}\pi(v,t)\cdot\din(v)&\le\left(\sum_{v\in V}\sum_{u\in\Nin(v)}\pi(v,t)\cdot\dout(u)\right)^{1/2}\left(\sum_{v\in V}\sum_{u\in\Nin(v)}\frac{\pi(v,t)}{\dout(u)}\right)^{1/2} \\
    &=\left(\sum_{v\in V}\pi(v,t)\sum_{u\in\Nin(v)}\dout(u)\right)^{1/2}\left(\sum_{u\in V}\frac{1}{\dout(u)}\sum_{v\in\Nout(u)}\pi(v,t)\right)^{1/2}.
\end{align*}
Note that $\sum_{u\in\Nin(v)}\dout(u)\le m$ and $\sum_{v\in V}\pi(v,t)=n\pi(t)$.
Also, the definition of PageRank implies that
\begin{align*}
    \pi(u,t)=\alpha\indicator{u=t}+\frac{1-\alpha}{\dout(u)}\sum_{v\in\Nout(u)}\pi(v,t),
\end{align*}
where $\indicator{}$ is the indicator function.
Consequently, we have
\begin{align} \label{ineqn:sum_PPR_children}
    \frac{1}{\dout(u)}\sum_{v\in\Nout(u)}\pi(v,t)=\frac{\pi(u,t)-\alpha\indicator{u=t}}{1-\alpha}\le\frac{\pi(u,t)}{1-\alpha}.
\end{align}
Combining these results yields
\begin{align*}
    \sum_{v\in V}\pi(v,t)\cdot\din(v)&\le\left(\sum_{v\in V}\pi(v,t)\cdot m\right)^{1/2}\left(\sum_{u\in V}\frac{\pi(u,t)}{1-\alpha}\right)^{1/2} \\
    &=\sqrt{n\pi(t)\cdot m\cdot n\pi(t)/(1-\alpha)}=O\big(n\pi(t)\cdot\sqrt{m}\big).
\end{align*}
Substituting this bound into $O\left(\sum_{v\in V}\pi(v,t)\cdot\din(v)/\rmax\right)$ gives the claimed bound of $O\big(n\pi(t)\cdot\sqrt{m}/\rmax\big)$.

For the bound involving $\Deltain$, we use the result in Theorem~\ref{thm:properties_bp} that the number of pushback operations is $O\big(n\pi(t)/\rmax\big)$.
Since each pushback operation takes $O(\Deltain)$ time, the $O\big(n\pi(t)/\rmax\cdot\Deltain\big)$ bound immediately follows.
For the bound involving $\Deltaout$, we use another complexity bound given in Theorem~\ref{thm:previous_cost_bp}, i.e., $O\left(\sum_{v\in V}\pi(v,t)\cdot\dout(v)/\rmax\right)$.
By directly bounding $\dout(v)$ by $\Deltaout$, the complexity bound becomes $O\left(\sum_{v\in V}\pi(v,t)\cdot\Deltaout/\rmax\right)=O\big(n\pi(t)/\rmax\cdot\Deltaout\big)$.
Combining these results gives the claimed bound of $O\left(n\pi(t)/\rmax\cdot\min\big(\Deltain,\Deltaout,\sqrt{m}\big)\right)$.
\end{proof}

Next, we use this new complexity bound of \bpush to tighten the upper bound of detecting the $\delta$-contributing set.
Recall that to this end, we need to find out all nodes $v$ with $\pi(v,t)\ge\delta n\pi(t)$.
For the easier case when the value of $n\pi(t)$ is known (as considered in \cite{andersen2008local}), we can simply set $\rmax=\delta n\pi(t)/2$, invoke $\bpush(t,\alpha,\rmax)$, and return the set $S=\big\{v\in V\mid\epib{v}\ge\rmax\big\}$.
Since $n\pi(t)/\rmax=\Theta(1/\delta)$, by Theorem~\ref{thm:cost_bp}, this algorithm runs in $O\left(\min\big(\Deltain,\Deltaout,\sqrt{m}\big)\big/\delta\right)$ time.
Additionally, by Theorem~\ref{thm:properties_bp}, we have $\pi(v,t)-\rmax\le\epib{v}\le\pi(v,t)$, and it is straightforward to check that $S$ contains the $\delta$-contributing set of $t$.
In fact, this algorithm also guarantees that $S$ is a subset of the $(\delta/2)$-contributing set of $t$.

For the more general case when no information other than $t$ and $\delta$ is available, our Theorem~\ref{thm:contribution_upper_bound} states that we can still achieve the $O\left(\min\big(\Deltain/\delta,\Deltaout/\delta,\sqrt{m}/\delta,m\big)\right)$ upper bound for detecting the $\delta$-contributing set of $t$.
At a high level, we achieve this by using the classic \textit{doubling} technique (or \textit{halving} in our case) to try a sequence of settings for $\rmax$, and in the meantime, recording some quantities involving the number of pushback operations performed on each node to determine when to terminate.
The formal \hyperlink{proof:contribution_upper_bound}{proof} of this result is nontrivial and more technical, which is given in Appendix~\ref{sec:deferred_proofs}.

%% file: SNPR_upper_bound.tex
\section{Revisiting the Complexity of \BiPPR} \label{sec:SNPR_upper_bound}

This section proves Theorem~\ref{thm:SNPR_upper_bound}, giving tighter upper bounds for single-node PageRank estimation as $O\left(n^{1/2}\cdot\min\left(\Deltain^{1/2},\Deltaout^{1/2},m^{1/4}\right)\right)$.
Recall that the task is to compute a multiplicative $(1\pm c)$-approximation of $\pi(t)$ w.p. at least $(1-\pf)$ for a target node $t$.
Our algorithm is basically a simple combination of \bpush and Monte Carlo simulation, which has been proposed in \cite{lofgren2016personalized} and dubbed \BiPPR.
However, \cite{lofgren2016personalized} only gives an average-case complexity bound of \BiPPR.
In the following, we reanalyze \BiPPR to derive the tight worst-case complexity bound.
Our analyses are significantly simpler than that in \cite{bressan2018sublinear,bressan2023sublinear} and even simpler than the original analyses of \BiPPR~\cite{lofgren2016personalized}.
Our analyses also illuminate why the complex techniques in \cite{bressan2018sublinear,bressan2023sublinear} are unnecessary, as discussed at the end of this section.

In this section, following \cite{bressan2018sublinear,bressan2023sublinear}, we assume that $n$ is known for simplicity.
In fact, one can estimate $n$ using the $\jump()$ query in $O\big(\sqrt{n}\big)$ time~\cite{bressan2015simple}, which does not affect the resultant upper bounds.

\header{\textbf{Monte Carlo Sampling.}}
First, we need a primitive to sample a node $v\in V$ w.p. $\pi(v)$.
This can be done by straightforwardly using the $\jump()$, $\outdeg(\cdot)$, and $\child(\cdot,\cdot)$ queries to simulate an $\alpha$-discounted random walk on $G$ and obtain the node at which it terminates.
We name this subroutine $\SampleNode()$ and define the indicator variable $\chi_v$ to be $\chi_v=\indicator{\SampleNode()=v}$.
Algorithm~\ref{alg:SampleNode} in Appendix~\ref{sec:table_notations} gives a pseudocode for $\SampleNode()$.
It is well-known that $\chi_v$ is a Bernoulli random variable that takes on value $1$ w.p. $\pi(v)$ and that $\SampleNode()$ takes expected $\Theta(1)$ time~\cite{fogaras2005towards,avrachenkov2007monte}.
Additionally, the $\chi_v$'s are \textit{negatively correlated} random variables~\cite{bressan2018sublinear,bressan2023sublinear}.
These notions and notations are consistent with those presented in \cite{bressan2018sublinear,bressan2023sublinear}.

\header{\textbf{Bidirectional Estimator.}}
We describe \BiPPR as a subroutine that takes as input parameters $\rmax$ and $n_r$ and outputs $\epi(t)$ as an estimate for $\pi(t)$.
\BiPPR first runs $\bpush(t,\alpha,\rmax)$ and uses its invariant (see Lemma~\ref{lem:invariant_backward}) to conceptually construct a \textit{bidirectional estimator} of $\pi(t)$ as a linear combination of the $\chi_v$'s.
Specifically, the bidirectional estimator $q(t)$ is defined as
\begin{align} \label{eqn:bidirectional}
    q(t)=\frac{1}{n}\sum_{v\in V}\epib{v}+\sum_{v\in V}\chi_v\cdot\rb{v}.
\end{align}
By Equation~\eqref{eqn:invariant_backward} and $\pi(t)=\frac{1}{n}\sum_{v\in V}\pi(v,t)$, one can verify that
\begin{align} \label{eqn:invariant_pi}
    \pi(t)=\frac{1}{n}\sum_{v\in V}\epib{v}+\sum_{v\in V}\pi(v)\cdot\rb{v},
\end{align}
so $q(t)$ is an unbiased estimator for $\pi(t)$.
Recall that \bpush reduces the residues so that $\rb{v}\le\rmax$ for any $v\in V$.
Thus, as $q(t)$ is a weighted sum of several negatively correlated random variables with small coefficients, it is a low-variance estimator.
\BiPPR then invokes $\SampleNode()$ for $n_r$ times to obtain several independent realizations of $q(t)$ and takes their average as the final estimate $\epi(t)$.

Now we give several lemmas for the proof of Theorem~\ref{thm:SNPR_upper_bound}.
First, by Theorem~\ref{thm:cost_bp} and the expected cost of $\SampleNode()$, we immediately have the following bound for the expected computational complexity of \BiPPR.

\begin{lemma} \label{lem:cost_BiPPR}
    The expected computational complexity of running \BiPPR for target node $t$ with parameters $\rmax\in(0,1]$ and $n_r$ is $O\left(n\pi(t)/\rmax\cdot\min\big(\Deltain,\Deltaout,\sqrt{m}\big)+n_r\right)$.
\end{lemma}

\noindent
Next, the following lemma upper bounds the variance of the estimator $\epi(t)$ returned by \BiPPR.

\begin{lemma} \label{lem:variance_BiPPR}
    The variance of the estimator $\epi(t)$ in \BiPPR can be bounded as:
    \begin{align*}
        \Var\big[\epi(t)\big]\le\frac{\rmax}{n_r}\cdot\pi(t).
    \end{align*}
\end{lemma}

\begin{proof}
Since \BiPPR computes $\epi(t)$ as the average of $n_r$ independent realizations of the bidirectional estimator $q(t)$, it suffices to show that $\Var\big[q(t)\big]\le\rmax\pi(t)$.
By the definition of $q(t)$ (see Equation~\eqref{eqn:bidirectional}), we have
\begin{align*}
    \Var\big[q(t)\big]=\Var\left[\frac{1}{n}\sum_{v\in V}\epib{v}+\sum_{v\in V}\chi_v\cdot\rb{v}\right]=\Var\left[\sum_{v\in V}\chi_v\cdot\rb{v}\right]
\end{align*}
since the randomness of $q(t)$ only comes from the $\chi_v$'s.
As the $\chi_v$'s are negatively correlated random variables, we further have
\begin{align*}
    \Var\big[q(t)\big]\le\sum_{v\in V}\Var\big[\chi_v\cdot\rb{v}\big]=\sum_{v\in V}\Var[\chi_v]\cdot\big(\rb{v}\big)^2.
\end{align*}
Note that the variance of the Bernoulli random variable $\chi_v$ is $\Var[\chi_v]=\pi(v)\big(1-\pi(v)\big)\le\pi(v)$, which further yields $\Var\big[q(t)\big]\le\sum_{v\in V}\pi(v)\cdot\big(\rb{v}\big)^2$.
Finally, recall that \bpush guarantees that $\rb{v}\le\rmax$ for all $v\in V$ and $\sum_{v\in V}\pi(v)\cdot\rb{v}=\pi(t)-1/n\cdot\sum_{v\in V}\epib{v}\le\pi(t)$ (see Equation~\eqref{eqn:invariant_pi}).
It follows that
\begin{align*}
    \Var\big[q(t)\big]\le\sum_{v\in V}\pi(v)\cdot\big(\rb{v}\big)^2\le\rmax\sum_{v\in V}\pi(v)\cdot\rb{v}\le\rmax\pi(t),
\end{align*}
which finishes the proof.
\end{proof}

\noindent
Up to this point, let us simplify the problem by applying the classic \textit{median trick}~\cite{jerrum1986random}.
The median trick tells us that if we have an algorithm that returns a satisfactory estimate w.p. at least $2/3$, then we can boost its success probability to at least $(1-\pf)$ by running it independently for $\Theta\big(\log(1/\pf)\big)$ times and taking the median of the results as the final answer.
Therefore, we can first consider the upper bound of \BiPPR when the failure probability parameter is $1/3$.
Then, in our case, since we only need to rerun the Monte Carlo phase of \BiPPR for $O\big(\log(1/\pf)\big)$ times and this cost can be balanced with the \bpush phase, we only need to multiply by $O\left(\big(\log(1/\pf)\big)^{1/2}\right)$ to yield the upper bound of \BiPPR for a general $\pf$ (see the \hyperlink{proof:SNPR_upper_bound}{proof} for Theorem~\ref{thm:SNPR_upper_bound} below).
Using the upper bound for $\Var\big[\epi(t)\big]$ and Chebyshev's inequality, Lemma~\ref{lem:BiPPR_requirement} below establishes a relationship between $\rmax$ and $n_r$ as a sufficient condition for \BiPPR to yield a satisfactory result w.p. at least $2/3$.

\begin{lemma} \label{lem:BiPPR_requirement}
    If the parameters of \BiPPR satisfy $n_r\ge3\rmax/\left(c^2\pi(t)\right)$, then w.p. at least $2/3$ the estimate $\epi(t)$ returned by \BiPPR is a multiplicative $(1\pm c)$-approximation of $\pi(t)$.
\end{lemma}

\begin{proof}
By Chebyshev's inequality and Lemma~\ref{lem:variance_BiPPR}, $n_r\ge3\rmax/\left(c^2\pi(t)\right)$ implies that
\begin{align*}
    \Pr\Big\{\big|\epi(t)-\pi(t)\big|\ge c\pi(t)\Big\}\le\frac{\Var\big[\epi(t)\big]}{c^2\big(\pi(t)\big)^2}\le\frac{\rmax\pi(t)}{n_r\cdot c^2\big(\pi(t)\big)^2}\le\frac{\rmax}{c^2\pi(t)}\cdot\frac{c^2\pi(t)}{3\rmax}=\frac{1}{3},
\end{align*}
finishing the proof.
\end{proof}

\noindent
Now it remains to prove Theorem~\ref{thm:SNPR_upper_bound} by striking a balance between the two phases of \BiPPR.
However, as we cannot figure out the desired setting of $\rmax$ and $n_r$ directly, we must employ an \textit{adaptive} setting of them to achieve the optimal upper bound, as explained in the following proof.

\begin{proof}[Proof of Theorem~\ref{thm:SNPR_upper_bound}] \hypertarget{proof:SNPR_upper_bound}
By Lemmas~\ref{lem:cost_BiPPR} and \ref{lem:BiPPR_requirement}, ideally, setting
\begin{align}
    \nonumber \rmax&=\Theta\left(n^{1/2}\cdot\min\left(\Deltain^{1/2},\Deltaout^{1/2},m^{1/4}\right)\pi(t)\cdot c\big(\log(1/\pf)\big)^{-1/2}\right) \\
    \text{and}\quad n_r&=\Theta\left(n^{1/2}\cdot\min\left(\Deltain^{1/2},\Deltaout^{1/2},m^{1/4}\right)\cdot c^{-1}\big(\log(1/\pf)\big)^{1/2}\right) \label{eqn:setting_nr}
\end{align}
could balance (the upper bounds of) the computational cost of \bpush and Monte Carlo simulations in \BiPPR while achieving the approximation requirement.
Here, we take into consideration the extra $O\big(\log(1/\pf)\big)$ factor introduced by the median trick.
A subtlety here is that we may set $\rmax>1$, but this does not affect our result since the cost of \bpush simply becomes $\Theta(1)$ for $\rmax>1$.
This ideal setting yields an upper bound of
\begin{align*}
    O\left(n^{1/2}\cdot\min\left(\Deltain^{1/2},\Deltaout^{1/2},m^{1/4}\right)\cdot c^{-1}\big(\log(1/\pf)\big)^{1/2}\right).
\end{align*}

However, we cannot directly set $\rmax$ and $n_r$ in this way since we do not know the parameters of $G$ and the value of $\pi(t)$, where $\pi(t)$ is exactly what we aim to estimate.
To achieve the desired upper bound, we once again resort to the doubling (halving) technique.

Let us first consider the scenario when we know the parameters of $G$, namely, $\Deltain$, $\Deltaout$, and $m$.
In this case, we can devise a workaround easily: we set $n_r$ according to Equation~\eqref{eqn:setting_nr} and repeatedly run \bpush with $\rmax=1,1/2,1/4,1/8,...$ until the total cost is about to surpass $\Theta(n_r)$, in which case we double $\rmax$ and restore the results of the last running of \bpush.
By doing so, the final $\rmax$ must be no larger than the desired setting, thereby ensuring that the error guarantee is met.
This approach can be readily made oblivious to the parameters of $G$ by imposing a doubling budget on the algorithm, using an argument similar to \cite[Section~5.9]{bressan2023sublinear}.
We set an initial budget for the algorithm and run the above process on this budget.
If we discover that $n_r\cdot\epi(t)=\Omega\left(c^{-2}\rmax\right)$ holds, then by standard concentration bounds (see, e.g., \cite{dagum2000optimal}), we can safely terminate and return $\epi(t)$ as a multiplicative $(1\pm c)$-approximation of $\pi(t)$.
Otherwise, we double the budget and repeat the process.
\end{proof}

\header{\textbf{Remark.}}
Let us compare our algorithm and analyses for single-node PageRank estimation to those in \cite{bressan2018sublinear,bressan2023sublinear}.
The algorithms in \cite{bressan2018sublinear,bressan2023sublinear} also construct an estimator of $\pi(t)$ as a linear combination of the Monte Carlo estimators, in a similar form to the bidirectional estimator in \BiPPR.
By exploring an induced subgraph $H$ of $G$ containing $t$, they construct a \textit{subgraph estimator} $q_H(t)$ of the form
\begin{align*}
    q_H(t)=c_H+\sum_{v\in V}\chi_v\cdot c_H(v),
\end{align*}
where $c_H$ and $c_H(v)$ are coefficients that depend on the structure of $H$.
The bidirectional estimator $q(t)=\frac{1}{n}\sum_{v\in V}\epib{v}+\sum_{v\in V}\chi_v\cdot\rb{v}$ in \BiPPR can be seen as a special form of the subgraph estimator.
The algorithms in \cite{bressan2018sublinear,bressan2023sublinear} are motivated by the problem that the nonzero coefficients of the $\chi_v$'s in the expression of $q_H(t)$ can be unbalanced, making it hard to guarantee that $q_H(t)$ is well-concentrated.
Therefore, a complex \textit{perfect weighted estimator} is designed as a weighted sum of several subgraph estimators to balance the coefficients of the $\chi_v$'s.
However, our analyses show that the bidirectional estimator does not suffer from this stated problem.
In the bidirectional estimator $q(t)=\frac{1}{n}\sum_{v\in V}\epib{v}+\sum_{v\in V}\chi_v\cdot\rb{v}$, even if the coefficients $r(v)$ are unbalanced, the property of \bpush guarantees that $\rb{v}\le\rmax$ for each $v$, and thus the variance of $q(t)$ is small enough to obtain the desired concentration property.
The rationale behind this point is intuitive: to develop an effective estimator of the form $q_H(t)=c_H+\sum_{v\in V}\chi_v\cdot c_H(v)$, the more straightforward approach is to make the coefficients $c_H(v)$ to be small enough instead of forcing them to be balanced.
On the other hand, \bpush serves as an efficient method to construct the bidirectional estimator $q(t)$.
Since the coefficients in the expression of $q(t)$ are exactly the reserves and residues computed by \bpush, the iterative methods used in \cite{bressan2018sublinear,bressan2023sublinear} to approximate the coefficients are unnecessary.
Also, by our analyses, the complexity of \bpush and the adaptive setting of the parameters are sufficient for \BiPPR to achieve the optimal complexity, eliminating the necessity of using the \textit{blacklisting} technique in \cite{bressan2018sublinear,bressan2023sublinear} to bypass nodes with large in-degrees.

%% file: contribution_lower_bound.tex
\section{Lower Bounds for Detecting the Contributing Set} \label{sec:contribution_lower_bound}

This section proves Theorem~\ref{thm:contribution_lower_bound_informal}, giving the lower bound of $\Omega\left(\min\big(\Deltain/\delta,\Deltaout/\delta,\sqrt{m}/\delta,m\big)\right)$ for detecting the $\delta$-contributing set, which matches our upper bound given in Theorem~\ref{thm:contribution_upper_bound}.
We emphasize that while we derive our upper bounds for this problem using a deterministic algorithm, our lower bound applies to the expected number of queries made by any randomized algorithm, even if the algorithm is allowed to fail with an arbitrarily large constant probability.
This proof also forms the foundation for proving the subsequent Theorems~\ref{thm:SNPR_lower_bound_informal} and \ref{thm:SNPR_lower_bound_Delta_in_informal}.
Now we formally state Theorem~\ref{thm:contribution_lower_bound_informal} as Theorem~\ref{thm:contribution_lower_bound} and prove it.

\begin{theorem} \label{thm:contribution_lower_bound}
    Choose any integer $p\ge2$ and any functions $\Deltain(n),\Deltaout(n)\in\Omega(1)\cap O(n)$, $m(n)\in\Omega(n)\cap O\big(n\Deltain(n)\big)\cap O\big(n\Deltaout(n)\big)$, and $\delta(n)\in(0,1)$.
    Consider any (randomized) algorithm $\mathcal{A}(t,\delta)$ that w.p. at least $1/p$ outputs a node set containing the $\delta$-contributing set of $t$, where $\mathcal{A}$ can only query the graph oracle to access unseen nodes and edges in the underlying graph and $\mathcal{A}$ cannot use the $\jump()$ query.
    Then, for every sufficiently large $n$, there exists a graph $H$ such that:
    (i) $H$ contains $\Theta(n)$ nodes and $\Theta(m)$ edges, and its maximum in-degree and out-degree is $\Theta(\Deltain)$ and $\Theta(\Deltaout)$, resp.;
    (ii) $H$ contains a node $t$ such that $\mathcal{A}(t,\delta)$ requires $\Omega\left(\min\big(\Deltain/\delta,\Deltaout/\delta,\sqrt{m}/\delta,m\big)\right)$ queries in expectation.
\end{theorem}

\begin{proof}[Proof of Theorem~\ref{thm:contribution_lower_bound}] \hypertarget{proof:contribution_lower_bound}
We prove the result by constructing a graph $H$ for every sufficiently large $n$ as a hard instance for detecting the $\delta$-contributing set.
We will set the input distribution to be the uniform distribution over $\mathcal{H}$, where $\mathcal{H}$ is the set of all graphs isomorphic to $H$ (obtained by permuting the node labels in $H$).
In the following, we describe the overall structure of $H$, derive the lower bound for the expected query complexity of $\mathcal{A}$ over this input distribution, and then instantiate the parameters of $H$ according to the given functions $\Deltain$, $\Deltaout$, $m$, and $\delta$ to obtain the desired lower bounds.

\begin{figure}[ht]
  \centering
  \includegraphics[width=0.9\linewidth]{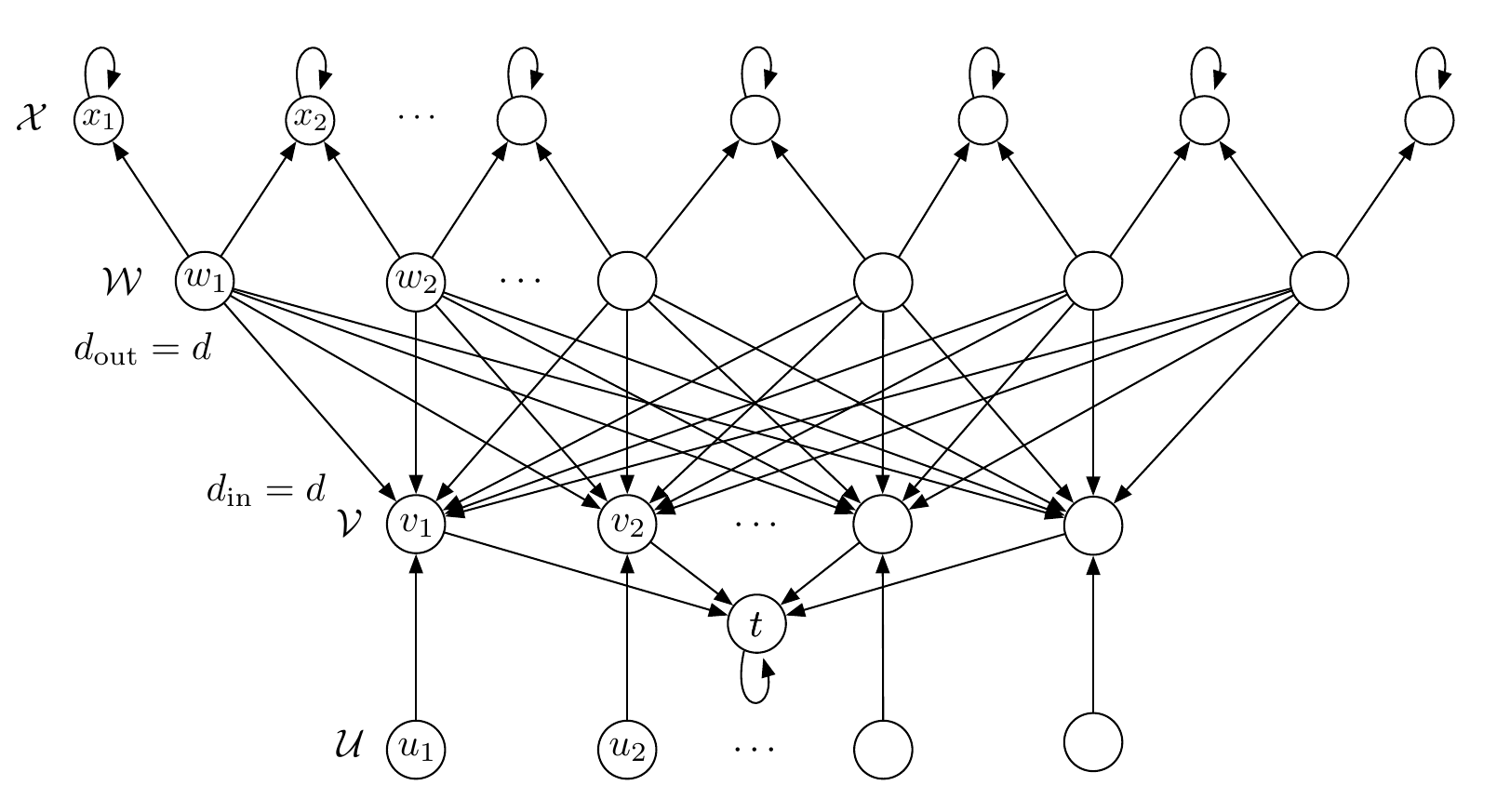}
  \caption{The graph $H$ as a hard instance for detecting the $\delta$-contributing set.} \label{fig:lower_bound_contribution}
\end{figure}

Figure~\ref{fig:lower_bound_contribution} depicts the graph $H$.
The node set of $H$ is composed of five disjoint node sets: $\{t\}$, $\mathcal{U}$, $\mathcal{V}$, $\mathcal{W}$, and $\mathcal{X}$.
We first set $\Nin(t)=\mathcal{V}\cup\{t\}$ and $|\mathcal{U}|=|\mathcal{V}|$, and add an outgoing edge from each node in $\mathcal{U}$ to an exclusive node in $\mathcal{V}$.
Next is the core part of our construction: we add edges from $\mathcal{W}$ to $\mathcal{V}$, making each node in $\mathcal{V}$ has in-degree $d$.
We also enforce that each node in $\mathcal{W}$ has out-degree $d$, and we achieve this by arbitrarily adding outgoing edges from them to nodes in $\mathcal{X}$ if needed, where $|\mathcal{X}|$ is set to be $n$ and each node in $\mathcal{X}$ has a self-loop.
We set $|\mathcal{W}|=\max\big(d,|\mathcal{V}|\big)$ to guarantee that this construction can be done: specifically, $|\mathcal{W}|\ge d$ guarantees that all nodes in $\mathcal{V}$ can receive $d$ incoming edges from different nodes in $\mathcal{W}$, and $|\mathcal{W}|\ge|\mathcal{V}|$ guarantees that the number of outgoing edges from $\mathcal{W}$ suffices to provide the $d|\mathcal{V}|$ incoming edges to $\mathcal{V}$.
We remark that if $|\mathcal{W}|=|\mathcal{V}|$, $\mathcal{X}$ will not receive incoming edges from $\mathcal{W}$.
Finally, we arbitrarily add to $H$ an isolated subgraph with $\Theta(n)$ nodes and $\Theta(m)$ edges, where the maximum in-degree and out-degree in the subgraph is $\Theta(\Deltain)$ and $\Theta(\Deltaout)$, resp.
This isolated subgraph is omitted in the figure.

The idea behind our construction of $H$ is to guarantee that $\mathcal{U}$ is a subset of the $\delta$-contributing set of $t$ and is difficult for $\mathcal{A}$ to detect.
In the following, we first detail the parameter constraints necessary for $H$ to exhibit the desired properties, i.e., to guarantee that $\mathcal{U}$ is a subset of the $\delta$-contributing set of $t$ and that $H$ meets the requirement (i) in the theorem statement.
Subsequently, we establish the difficulty of detecting $\mathcal{U}$ by giving a lower bound on the number of queries required for detecting it, which is an expression of the parameters of $H$.

\header{\textbf{Constraints on the parameters of $\boldsymbol{H}$.}}
We first calculate $t$'s PageRank score and the PPR value from a node $u\in\mathcal{U}$ to $t$, namely, $\pi(t)$ and $\pi(u,t)$.
By Equation~\eqref{eqn:iterative_pagerank}, nodes in $\mathcal{W}$ have PageRank scores $\Theta(1/n)$, nodes in $\mathcal{V}$ have PageRank scores $\Theta(1/n+d\cdot1/n/d)=\Theta(1/n)$, and $\pi(t)=\Theta\big(1/n\cdot|\mathcal{V}|\big)$.
We also have $\pi(u,t)=\Theta(1)$ by definition.
Consequently, $\pi(u,t)\big/\big(n\pi(t)\big)=\Theta\big(1/|\mathcal{V}|\big)$, and thus, as long as $|\mathcal{V}|=O(1/\delta)$, we can guarantee that $\pi(u,t)\ge\delta n\pi(t)$ by setting the hidden constants in the asymptotic notations properly.
As for ensuring that $H$ meets the requirement (i), as $H$ contains an isolated subgraph with proper parameters, we only need to ensure that the parameters of the main portion of $H$ do not surpass the requirements.
To sum up, the parameter constraints are:
\begin{enumerate}
  \item $|\mathcal{V}|=O(1/\delta)$ for guaranteeing the property of $\mathcal{U}$.
  \item The number of nodes $\Theta\left(n+|\mathcal{W}|\right)=\Theta\left(n+\max\big(d,|\mathcal{V}|\big)\right)$ is $O(n)$. \label{constraint:n}
  \item The number of edges $\Theta\big(d|\mathcal{W}|\big)=\Theta\big(d\cdot\max\big(d,|\mathcal{V}|\big)\big)$ is $O(m)$. \label{constraint:m}
  \item The maximum in-degree $\max\big(d,|\mathcal{V}|\big)$ is $O(\Deltain)$. \label{constraint:Deltain}
  \item The maximum out-degree $d=O(\Deltaout)$. \label{constraint:Deltaout}
\end{enumerate}

\header{\textbf{Lower bound.}}
Regarding the difficulty of detecting $\mathcal{U}$, we argue that when the underlying graph is chosen uniformly at random from $\mathcal{H}$, $\mathcal{A}$ must perform $\Omega\big(d|\mathcal{V}|\big)$ queries in expectation to detect all nodes in $\mathcal{U}$ w.p. at least $1/p$.
Recall that $\mathcal{A}$ can only explore unseen nodes through the graph oracle without using the $\textsc{jump}()$ query.
Thus, to detect a node $u\in\mathcal{U}$, $\mathcal{A}$ must call the corresponding $\textsc{parent}(v,\cdot)$ query that returns $u$, where $v$ is the child of $u$.
Note that there are $d$ queries of the form $\textsc{parent}(v,\cdot)$, each of which is equally likely to be the one relevant to $u$, and that $\mathcal{A}$ cannot learn the results of these queries without querying them.
Therefore, it must take $\Omega(d)$ queries in expectation for $\mathcal{A}$ to detect $u$ with a non-vanishing probability.
As this argument holds for each $u\in\mathcal{U}$, the desired lower bound of $\Omega\big(d|\mathcal{U}|\big)=\Omega\big(d|\mathcal{V}|\big)$ follows.

Now it remains to instantiate $|\mathcal{V}|$ and $d$ to derive the claimed $\Omega\left(\min\big(\Deltain/\delta,\Deltaout/\delta,\sqrt{m}/\delta,m\big)\right)$ lower bound.
We will give the settings of $|\mathcal{V}|$ and $d$ for different cases, and it would be straightforward to check that they meet the constraints above.
We first consider the easier case when $\sqrt{m}$ is the smallest among $\Deltain$, $\Deltaout$, and $\sqrt{m}$.

\header{\textbf{Case 1:} $\sqrt{m}=O(\Deltain)$, $\sqrt{m}=O(\Deltaout)$, and $\delta=\Omega\big(1/\sqrt{m}\big)$.}

We set $|\mathcal{V}|=\Theta(1/\delta)$ and $d=\Theta\big(\sqrt{m}\big)$, giving the lower bound of $\Omega\big(d|\mathcal{V}|\big)=\Omega\big(\sqrt{m}/\delta\big)=\Omega\left(\min\big(\Deltain/\delta,\Deltaout/\delta,\sqrt{m}/\delta,m\big)\right)$.

\header{\textbf{Case 2:} $\sqrt{m}=O(\Deltain)$, $\sqrt{m}=O(\Deltaout)$, and $\delta=O\big(1/\sqrt{m}\big)$.}

We set $|\mathcal{V}|=\Theta\big(\sqrt{m}\big)$ and $d=\Theta\big(\sqrt{m}\big)$, giving the lower bound of $\Omega\big(d|\mathcal{V}|\big)=\Omega(m)=\Omega\left(\min\big(\Deltain/\delta,\Deltaout/\delta,\sqrt{m}/\delta,m\big)\right)$.

For the more subtle cases where $\Deltain$ or $\Deltaout$ is the smallest, we need to modify the structure of $H$ slightly.
The problem is that node $t$ has in-degree $|\mathcal{V}|$ in Figure~\ref{fig:lower_bound_contribution}, which can be larger than the required $\Deltain$ if we set, say, $|\mathcal{V}|=\Theta(1/\delta)$ and $d=\Theta(\Deltaout)$.
To tackle this issue, we replace the edges from $\mathcal{V}$ to $t$ by the following \textit{multi-level structure}.

\header{\textbf{Multi-Level Structure.}}
Figure~\ref{fig:multi_level_structure} illustrates the simple multi-level structure.
Without loss of generality, we assume that all quantities involved below are integers.
The structure is akin to a reversed complete $\big(2/(1-\alpha)\big)$-ary tree: node $t$ is the ``root'' at the bottom, other nodes are partitioned into $L$ levels, and each node except the ``leaves'' at the top level has $2/(1-\alpha)$ distinct parents in the level above it.
We denote by $\mathcal{V}_i$ the set of nodes in level $i$ for $1\le i\le L$, counting from top to bottom.
Clearly, $L=\log_{2/(1-\alpha)}|\mathcal{V}_1|$, the number of nodes in the structure is $\Theta\big(|\mathcal{V}_1|\big)$, and the maximum in-degree and out-degree in the structure are both $\Theta(1)$.

\begin{figure}[ht]
\centering
\includegraphics[width=0.6\linewidth]{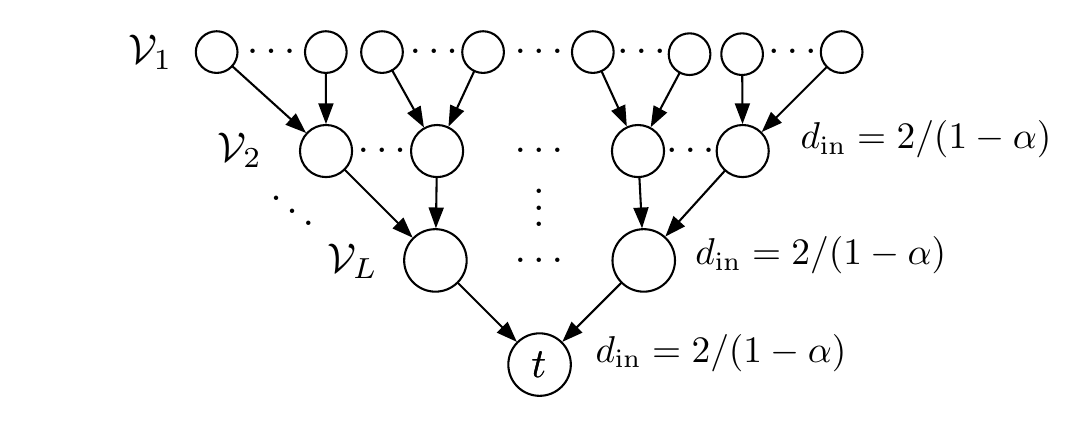}
\caption{The multi-level structure.} \label{fig:multi_level_structure}
\end{figure}

Now, we embed this multi-level structure into $H$, with $\mathcal{V}_1$ in Figure~\ref{fig:multi_level_structure} taking place of $\mathcal{V}$ in Figure~\ref{fig:lower_bound_contribution} (we will regard $\mathcal{V}=\mathcal{V}_1$).
Crucially, we need to analyze $t$'s PageRank score and the PPR score $\pi(v,t)$ for node $v\in\mathcal{V}$ after this modification.
By straightforward mathematical induction, we can prove that the nodes in $\mathcal{V}_i$ have PageRank scores $\Theta\left(2^{i-1}/n\right)$: for the induction step, we can derive that the nodes in $\mathcal{V}_{i+1}$ have PageRank scores $\Theta\left(1/n+2/(1-\alpha)\cdot2^{i-1}/n\cdot(1-\alpha)\right)=\Theta\left(2^i/n\right)$.
Hence, we have $\pi(t)=\Theta\left(1/n+2/(1-\alpha)\cdot2^{L-1}/n\cdot(1-\alpha)\right)=\Theta\left(2^L/n\right)$.
On the other hand, for each $v\in\mathcal{V}$, $\pi(v,t)=\Theta\left((1-\alpha)^L\right)$ and thus $\pi(v,t)\big/\big(n\pi(t)\big)=\Theta\left(\big(1-\alpha)/2\big)^L\right)=\Theta\big(1/|\mathcal{V}|\big)$.
Consequently, for each node $u\in\mathcal{U}$, we also have $\pi(u,t)\big/\big(n\pi(t)\big)=\Theta\big(1/|\mathcal{V}|\big)$, agreeing with the property in the original Figure~\ref{fig:lower_bound_contribution}.

In a word, the inserted multi-level structure takes exactly the same effect as the original edges from $\mathcal{V}$ to $t$ for our arguments, removing the constraint that $|\mathcal{V}|=O(\Deltain)$ without changing other properties of $H$.
Now we are ready to give the settings of $|\mathcal{V}|$ and $d$ for other cases.

\header{\textbf{Case 3:} $\Deltaout=O(\Deltain)$, $\Deltaout=O\big(\sqrt{m}\big)$, and $\delta=\Omega(\Deltaout/m).$}

We set $|\mathcal{V}|=\Theta(1/\delta)$ and $d=\Theta(\Deltaout)$, giving the lower bound of $\Omega\big(d|\mathcal{V}|\big)=\Omega(\Deltaout/\delta)=\Omega\left(\min\big(\Deltain/\delta,\Deltaout/\delta,\sqrt{m}/\delta,m\big)\right)$.
Note that as $m=O(n\Deltaout)$, we have $|\mathcal{V}|=O(m/\Deltaout)=O(n)$.

\header{\textbf{Case 4:} $\Deltaout=O(\Deltain)$, $\Deltaout=O\big(\sqrt{m}\big)$, and $\delta=O(\Deltaout/m).$}

We set $|\mathcal{V}|=\Theta(m/\Deltaout)$ and $d=\Theta(\Deltaout)$, giving the lower bound of $\Omega\big(d|\mathcal{V}|\big)=\Omega(m)=\Omega\left(\min\big(\Deltain/\delta,\Deltaout/\delta,\sqrt{m}/\delta,m\big)\right)$.

\header{\textbf{Case 5:} $\Deltain=O(\Deltaout)$, $\Deltain=O\big(\sqrt{m}\big)$, and $\delta=\Omega(\Deltain/m).$}

We set $|\mathcal{V}|=\Theta(1/\delta)$ and $d=\Theta(\Deltain)$, giving the lower bound of $\Omega\big(d|\mathcal{V}|\big)=\Omega(\Deltain/\delta)=\Omega\left(\min\big(\Deltain/\delta,\Deltaout/\delta,\sqrt{m}/\delta,m\big)\right)$.
Note that as $m=O(n\Deltain)$, we have $|\mathcal{V}|=O(m/\Deltain)=O(n)$.

\header{\textbf{Case 6:} $\Deltain=O(\Deltaout)$, $\Deltain=O\big(\sqrt{m}\big)$, and $\delta=O(\Deltain/m).$}

We set $|\mathcal{V}|=\Theta(m/\Deltain)$ and $d=\Theta(\Deltain)$, giving the lower bound of $\Omega\big(d|\mathcal{V}|\big)=\Omega(m)=\Omega\left(\min\big(\Deltain/\delta,\Deltaout/\delta,\sqrt{m}/\delta,m\big)\right)$.

As the six cases considered above cover all possible relationships among $\Deltain$, $\Deltaout$, $m$, and $\delta$, the proof is completed.
\end{proof}

\header{\textbf{Remark.}}
In summary, the difficulty of detecting the $\delta$-contributing set of $t$ in $H$ lies in detecting an unexplored parent (a node in $\mathcal{U}$) of a node in $\mathcal{V}$ among a number of its visited parents (nodes in $\mathcal{W}$).
This difficulty turns out to be a special property of the arc-centric graph-access model.
In contrast, under the node-centric graph-access model, our structure between $\mathcal{W}$ and $\mathcal{V}$ cannot increase the query complexity since $|\mathcal{V}|$ $\neigh(\cdot)$ queries to nodes in $\mathcal{V}$ suffice to reveal all nodes in $\mathcal{U}\cup\mathcal{V}\cup\mathcal{W}$.
Also, if the algorithm is allowed to access the parents of a node in some particular order, the difficulty of identifying the special parent may be dramatically reduced.
In fact, \RBS~\cite{wang2020personalized} is such an example: by pre-sorting the parents of each node in ascending order of their out-degrees using linear time and space, \RBS can effortlessly detect the special parents in $\mathcal{U}$ at query time, potentially breaking our lower bound under the arc-centric graph-access model.

%% file: SNPR_lower_bound.tex
\section{Lower Bounds for Single-Node PageRank Estimation} \label{sec:SNPR_lower_bound}

This section proves Theorems~\ref{thm:SNPR_lower_bound_informal} and \ref{thm:SNPR_lower_bound_Delta_in_informal}, which establish lower bounds for single-node PageRank estimation.
Theorem~\ref{thm:SNPR_lower_bound_informal} focuses on parameters $n$, $\Deltaout$, and $m$ and gives a clean lower bound of $\Omega\left(n^{1/2}\cdot\min\left(\Deltaout^{1/2},m^{1/4}\right)\right)$, while Theorem~\ref{thm:SNPR_lower_bound_Delta_in_informal} further considers $\Deltain$ and is more technical.
These two proofs share a similar framework with our \hyperlink{proof:contribution_lower_bound}{proof} of Theorem~\ref{thm:contribution_lower_bound}, and our arguments are partly inspired by the proof in \cite[Section~7]{bressan2023sublinear}.
We first give the formal version of Theorem~\ref{thm:SNPR_lower_bound_informal} as Theorem~\ref{thm:SNPR_lower_bound} and prove it.

\begin{theorem} \label{thm:SNPR_lower_bound}
    Choose any integer $p\ge2$ and any functions $\Deltaout(n)\in\Omega(1)\cap O(n)$ and $m(n)\in\Omega(n)\cap O\big(n\Deltaout(n)\big)$.
    Consider any (randomized) algorithm $\mathcal{B}(t)$ that w.p. at least $1/p$ estimates $\pi(t)$ within a multiplicative factor of $O(1)$, where $\mathcal{B}$ can only query the graph oracle to access unseen nodes and edges in the underlying graph.
    Then, for every sufficiently large $n$, there exists a graph $H$ such that:
    (i) $H$ contains $\Theta(n)$ nodes and $\Theta(m)$ edges, and its maximum out-degree is $\Theta(\Deltaout)$;
    (ii) $H$ contains a node $t$ such that $\mathcal{B}(t)$ requires $\Omega\left(n^{1/2}\cdot\min\left(\Deltaout^{1/2},m^{1/4}\right)\right)$ queries in expectation.
\end{theorem}

\begin{proof}[Proof of Theorem~\ref{thm:SNPR_lower_bound}] \hypertarget{proof:SNPR_lower_bound}
We prove the result by constructing a family of graphs $\{H_i\}_{i=0}^{p}$ for every sufficiently large $n$ as hard instances for estimating single-node PageRank.
These graphs $\{H_i\}$ are designed to be difficult for the algorithm $\mathcal{B}$ to distinguish, but $\mathcal{B}$ must distinguish among them to meet the approximation requirement for estimating $\pi(t)$.
$\{H_i\}$ are constructed based on the graph $H$ given in the \hyperlink{proof:contribution_lower_bound}{proof} of Theorem~\ref{thm:contribution_lower_bound}.
In this proof, as we do not consider $\Deltain$, we do not need to use the multi-level structure.
We will set the input distribution to be the uniform distribution over the graph set $\bigcup_{i=0}^{p}\mathcal{H}_i$, where $\mathcal{H}_i$ is the set of all graphs isomorphic to $H_i$.
In the following, we describe the overall structure of $\{H_i\}$, derive the lower bound for the expected query complexity of $\mathcal{B}$ over this input distribution, and then instantiate the parameters of $\{H_i\}$ according to the given functions $\Deltaout$ and $m$ to obtain the desired lower bounds.

We depict the graph $H_i$ (for $0\le i\le p$) in Figure~\ref{fig:lower_bound_SNPR}.
Basically, $H_i$ is obtained by adding to $H$ a node set $\mathcal{Y}$ and its relevant edges, where we set $|\mathcal{Y}|=\Theta\big(|\mathcal{V}|\big)$.
We arbitrarily designate a node in $\mathcal{U}$ as $u_{*}$ and denote its unique child in $\mathcal{V}$ as $v_{*}$.
In $H_i$, $i/p\cdot|\mathcal{Y}|$ nodes in $\mathcal{Y}$ have an outgoing edge to $u_{*}$ and other nodes in $\mathcal{Y}$ only have a self-loop (again, without loss of generality we assume that all quantities involved are integers).
Note that for different $i$, the only difference between $H_i$ lies in the portion of nodes in $\mathcal{Y}$ that point to $u_{*}$.

\begin{figure}[ht]
  \centering
  \includegraphics[width=0.9\linewidth]{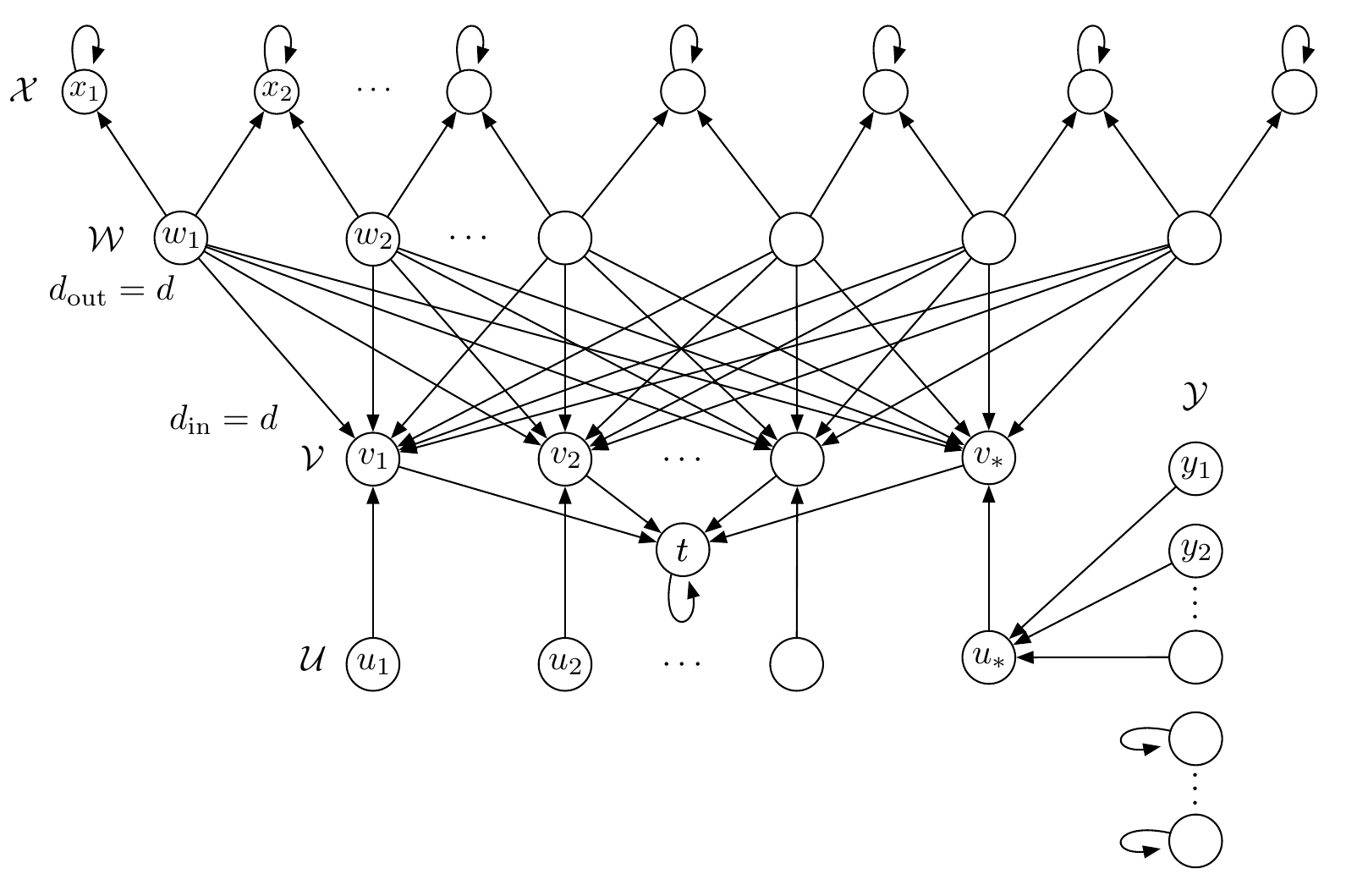}
  \caption{The graph $H_i$ for $0\le i\le p$ as hard instances for estimating single-node PageRank. Node $u_{*}$ has $i/p\cdot|\mathcal{Y}|$ parents in $\mathcal{Y}$.} \label{fig:lower_bound_SNPR}
\end{figure}

We first show that when the underlying graph is chosen uniformly at random from $\bigcup_{i=0}^{p}\mathcal{H}_i$, $\mathcal{B}$ must distinguish between the graphs $\{H_i\}_{i=0}^{p}$ to output an acceptable approximation of $\pi(t)$.
We denote the PageRank score of $t$ in $H_i$ as $\pi_i(t)$, and we aim to show that $\pi_i(t)=\big(1+\Omega(1)\big)\pi_{i-1}(t)$ for each $1\le i\le p$.
Recall that nodes in $\mathcal{V}\setminus\{v_{*}\}$ have PageRank score $\Theta(1/n)$.
We can compute that the PageRank score of $u_{*}$ is $\Theta\left(1/n\cdot\big(1+i/p\cdot|\mathcal{Y}|\big)\right)$, which is $O\big(|\mathcal{Y}|/n\big)=O\big(|\mathcal{V}|/n\big)$.
Thus, $v_{*}$ has PageRank score $O\big(|\mathcal{V}|/n\big)$ and $\pi_i(t)=\Theta\big(|\mathcal{V}|/n\big)$.
We also have $\pi_i(t)-\pi_{i-1}(t)=\Theta\big(1/n\cdot1/p\cdot|\mathcal{Y}|\big)=\Theta\big(|\mathcal{V}|/n\big)$, which implies $\pi_i(t)=\big(1+\Omega(1)\big)\pi_{i-1}(t)$ for each $1\le i\le p$.
Thus, if $\mathcal{B}$ does not distinguish between these graphs $\{H_i\}_{i=0}^{p}$, then the probability that it estimates $\pi(t)$ within a multiplicative factor of $O(1)$ is at most $1/(p+1)$, which is not acceptable.

Furthermore, by our construction, $\mathcal{B}$ must detect at least one node in $\{u_{*}\}\cup\mathcal{Y}$ to distinguish between $\{H_i\}$ since other parts of the graphs are identical.
Formally speaking, as long as no node in $\{u_{*}\}\cup\mathcal{Y}$ is detected by $\mathcal{B}$, conditioned on the outcome of the past queries the probability that the underlying graph is $H_i$ remains $1/(p+1)$ for each $0\le i\le p$.

Now we argue that $\mathcal{B}$ must perform $\Omega\left(\min\big(d|\mathcal{V}|,n/|\mathcal{V}|\big)\right)$ queries in expectation to detect at least one node in $\{u_{*}\}\cup\mathcal{Y}$ w.p. at least $1/p$.
Recall that $\mathcal{B}$ can only explore unseen nodes in the underlying graph through local queries or the global $\jump()$ query to the graph oracle.
Note that the output of the $\jump()$ query is independent of other operations made by $\mathcal{B}$.
To detect one node in $\{u_{*}\}\cup\mathcal{Y}$ w.p. at least $1/p$ using the $\jump()$ query, $\mathcal{B}$ must perform $\Omega\big(n/|\mathcal{Y}|\big)=\Omega\big(n/|\mathcal{V}|\big)$ $\textsc{jump}()$ queries in expectation.

If $\mathcal{B}$ fails to detect one of these nodes using the $\jump()$ query, then the only way for $\mathcal{B}$ to access them is to invoke the particular $\parent(v_{*},\cdot)$ query that returns $u_{*}$.
Now consider all $\Theta\big(d|\mathcal{V}|\big)$ different $\parent(\cdot,\cdot)$ queries relevant to the nodes in $\mathcal{V}$.
Since the node labels are uniformly randomly permuted in our input distribution, each of these queries is equally likely to be the query that returns $u_{*}$.
Thus, $\mathcal{B}$ needs $\Omega\big(d|\mathcal{V}|\big)$ $\parent(\cdot,\cdot)$ queries in expectation to detect $u_{*}$ with a non-vanishing probability.
Although finding a $\parent(\cdot,\cdot)$ query that connects $\mathcal{V}$ and $\mathcal{U}$ could help $\mathcal{B}$ to determine whether the returned node is the node $u_{*}$, this lower bound still holds since identifying such a query still needs $\Omega(d)$ time in expectation.
Combining the two bounds above yields the claimed lower bound of $\Omega\left(\min\big(d|\mathcal{V}|,n/|\mathcal{V}|\big)\right)$.

Now it remains to instantiate $|\mathcal{V}|$ and $d$ to derive the claimed $\Omega\left(n^{1/2}\cdot\min\left(\Deltaout^{1/2},m^{1/4}\right)\right)$ lower bound.
We consider two cases depending on the relationship between $\Deltaout$ and $\sqrt{m}$.
The settings of $|\mathcal{V}|$ and $d$ should satisfy Constraints~\ref{constraint:n}, \ref{constraint:m}, and \ref{constraint:Deltaout} given in the \hyperlink{proof:contribution_lower_bound}{proof} of Theorem~\ref{thm:contribution_lower_bound}, which is straightforward to verify.

\header{\textbf{Case 1:} $m^{1/2}=O(\Deltaout)$.}

We set $|\mathcal{V}|=\Theta\left(n^{1/2}m^{-1/4}\right),d=\Theta\left(m^{1/2}\right)$, giving the lower bound of $\Omega\left(\min\big(d|\mathcal{V}|,n/|\mathcal{V}|\big)\right)=\Omega\left(n^{1/2}m^{1/4}\right)=\Omega\left(n^{1/2}\cdot\min\left(\Deltaout^{1/2},m^{1/4}\right)\right)$.

\header{\textbf{Case 2:} $\Deltaout=O\left(m^{1/2}\right)$.}

We set $|\mathcal{V}|=\Theta\left(n^{1/2}\Deltaout^{-1/2}\right),d=\Theta(\Deltaout)$, giving the lower bound of $\Omega\left(\min\big(d|\mathcal{V}|,n/|\mathcal{V}|\big)\right)=\Omega\left(n^{1/2}\Deltaout^{1/2}\right)=\Omega\left(n^{1/2}\cdot\min\left(\Deltaout^{1/2},m^{1/4}\right)\right)$.

The proof is finished by combining these two cases.
\end{proof}

\header{\textbf{Remark.}}
Similarly, the difficulty of our hard instances here comes from the structure between $\mathcal{W}$ and $\mathcal{V}$ and is a property of the arc-centric graph-access model.
In comparison, the instances given in \cite{bressan2018sublinear,bressan2023sublinear} do not encompass this two-level structure and the decisive node therein (corresponding to our $u_{*}$) is a direct parent of the target node $t$, which is easier to detect.

Next, we focus on Theorem~\ref{thm:SNPR_lower_bound_Delta_in_informal}.
As Theorem~\ref{thm:SNPR_lower_bound_Delta_in_informal} further integrates $\Deltain$ as a parameter, it introduces some technicalities because, in Figure~\ref{fig:lower_bound_SNPR}, node $t$ and $u_{*}$ may have large in-degrees.
When $\min(\Deltain,\Deltaout)=\Omega\left(n^{1/3}\right)$, our previous construction of $\{H_i\}$ still works and we are able to prove a clean lower bound of $\Omega\left(n^{1/2}\cdot\min\left(\Deltain^{1/2},\Deltaout^{1/2},m^{1/4}\right)\right)$.
When $\min(\Deltain,\Deltaout)=o\left(n^{1/3}\right)$, again we need to embed the multi-level structure (see Figure~\ref{fig:multi_level_structure} in the \hyperlink{proof:contribution_lower_bound}{proof} of Theorem~\ref{thm:contribution_lower_bound}) into the constructed graphs $\{H_i\}$ to diminish their maximum in-degrees.
However, for single-node PageRank estimation, this modification cannot fully preserve the properties of the graphs $\{H_i\}$ and we prove a slightly weaker lower bound of $\Omega\left(n^{1/2-\smallexpo}\big(\min(\Deltain,\Deltaout)\big)^{1/2+\smallexpo}\right)$ for any constant $\smallexpo>0$.
We formally state Theorem~\ref{thm:SNPR_lower_bound_Delta_in_informal} as Theorem~\ref{thm:SNPR_lower_bound_Delta_in} and prove it.

\begin{theorem} \label{thm:SNPR_lower_bound_Delta_in}
    Choose any integer $p\ge2$ and any functions $\Deltain(n),\Deltaout(n)\in\Omega(1)\cap O(n)$ and $m(n)\in\Omega(n)\cap O\big(n\Deltain(n)\big)\cap O\big(n\Deltaout(n)\big)$.
    Consider any (randomized) algorithm $\mathcal{B}(t)$ that w.p. at least $1/p$ estimates $\pi(t)$ within a multiplicative factor of $O(1)$, where $\mathcal{B}$ can only query the graph oracle to access unseen nodes and edges in the underlying graph.

    If $\min(\Deltain,\Deltaout)=\Omega\left(n^{1/3}\right)$, then for every sufficiently large $n$, there exists a graph $H$ such that:
    (i) $H$ contains $\Theta(n)$ nodes and $\Theta(m)$ edges, and its maximum in-degree and out-degree is $\Theta(\Deltain)$ and $\Theta(\Deltaout)$, resp.;
    (ii) $H$ contains a node $t$ such that $\mathcal{B}(t)$ requires $\Omega\left(n^{1/2}\cdot\min\left(\Deltain^{1/2},\Deltaout^{1/2},m^{1/4}\right)\right)$ queries in expectation.

    If instead $\min(\Deltain,\Deltaout)=o\left(n^{1/3}\right)$, then for any constant real number $\smallexpo>0$ and every sufficiently large $n$, there exists a graph $H$ such that:
    (i) $H$ contains $\Theta(n)$ nodes and $\Theta(m)$ edges, and its maximum in-degree and out-degree is $\Theta(\Deltain)$ and $\Theta(\Deltaout)$, resp.;
    (ii) $H$ contains a node $t$ such that $\mathcal{B}(t)$ requires $\Omega\left(n^{1/2-\smallexpo}\big(\min(\Deltain,\Deltaout)\big)^{1/2+\smallexpo}\right)$ queries in expectation.
\end{theorem}

\begin{proof}[Proof of Theorem~\ref{thm:SNPR_lower_bound_Delta_in}] \hypertarget{proof:SNPR_lower_bound_Delta_in}
We first examine the easier case when $\min(\Deltain,\Deltaout)=\Omega\left(n^{1/3}\right)$.
In this case, we can use the same construction and arguments as in the \hyperlink{proof:SNPR_lower_bound}{proof} of Theorem~\ref{thm:SNPR_lower_bound}, save that our settings for $|\mathcal{V}|$ and $d$ should additionally satisfy Constraint~\ref{constraint:Deltain} in the \hyperlink{proof:contribution_lower_bound}{proof} of Theorem~\ref{thm:contribution_lower_bound}, i.e., $d=O(\Deltain)$ and $|\mathcal{V}|=O(\Deltain)$.
We give our settings as follows.

\header{\textbf{Case 1.1:} $\min(\Deltain,\Deltaout)=\Omega\left(n^{1/3}\right)$ and $m^{1/2}=O\big(\min(\Deltain,\Deltaout)\big)$.}

We set $|\mathcal{V}|=\Theta\left(n^{1/2}m^{-1/4}\right),d=\Theta\left(m^{1/2}\right)$, giving the lower bound of $\Omega\left(\min\big(d|\mathcal{V}|,n/|\mathcal{V}|\big)\right)=\Omega\left(n^{1/2}m^{1/4}\right)=\Omega\left(n^{1/2}\cdot\min\left(\Deltain^{1/2},\Deltaout^{1/2},m^{1/4}\right)\right)$.

\header{\textbf{Case 1.2:} $\min(\Deltain,\Deltaout)=\Omega\left(n^{1/3}\right)$ and $\Deltaout=O\left(\min\left(\Deltain,m^{1/2}\right)\right)$.}

We set $|\mathcal{V}|=\Theta\left(n^{1/2}\Deltaout^{-1/2}\right),d=\Theta(\Deltaout)$, giving the lower bound of $\Omega\left(\min\big(d|\mathcal{V}|,n/|\mathcal{V}|\big)\right)=\Omega\left(n^{1/2}\Deltaout^{1/2}\right)=\Omega\left(n^{1/2}\cdot\min\left(\Deltain^{1/2},\Deltaout^{1/2},m^{1/4}\right)\right)$.
Note that we have $|\mathcal{V}|=O\left(n^{1/3}\right)=O(\Deltain)$ since $\min(\Deltain,\Deltaout)=\Omega\left(n^{1/3}\right)$.

\header{\textbf{Case 1.3:} $\min(\Deltain,\Deltaout)=\Omega\left(n^{1/3}\right)$ and $\Deltain=O\left(\min\left(\Deltaout,m^{1/2}\right)\right)$.}

We set $|\mathcal{V}|=\Theta\left(n^{1/2}\Deltain^{-1/2}\right)$ and $d=\Theta(\Deltain)$, giving the lower bound of $\Omega\left(\min\big(d|\mathcal{V}|,n/|\mathcal{V}|\big)\right)=\Omega\left(n^{1/2}\Deltain^{1/2}\right)=\Omega\left(n^{1/2}\cdot\min\left(\Deltain^{1/2},\Deltaout^{1/2},m^{1/4}\right)\right)$.
Note that we have $|\mathcal{V}|=O\left(n^{1/3}\right)=O(\Deltain)$ since $\Deltain=\Omega\left(n^{1/3}\right)$.

We turn now to the more complicated case when $\min(\Deltain,\Deltaout)=o\left(n^{1/3}\right)$.
Note that this implies $\min(\Deltain,\Deltaout)=O\left(m^{1/2}\right)$.
In this case, we prove that for any constant $\smallexpo>0$, the lower bound of $\Omega\left(n^{1/2-\smallexpo}\big(\min(\Deltain,\Deltaout)\big)^{1/2+\smallexpo}\right)$ holds.

Following the idea in the \hyperlink{proof:contribution_lower_bound}{proof} of Theorem~\ref{thm:contribution_lower_bound}, to reduce the in-degree of $t$ and $u_{*}$ in Figure~\ref{fig:lower_bound_SNPR}, we embed two multi-level structures (see Figure~\ref{fig:multi_level_structure}) into $\{H_i\}$, replacing the edges from $\mathcal{V}$ to $t$ and the edges from $\mathcal{Y}$ to $u_{*}$.
Unlike Figure~\ref{fig:multi_level_structure}, here we set the in-degrees of the nodes in the structure (except at the top level) to be $k/(1-\alpha)$, where $k$ is a constant positive integer to be determined by $\alpha$ and $\smallexpo$.
We will set $|\mathcal{V}|$ and $|\mathcal{Y}|$ differently, and we denote the number of levels in the corresponding multi-level structures by $L_{\mathcal{V}}=\log_{k/(1-\alpha)}|\mathcal{V}|$ and $L_{\mathcal{Y}}=\log_{k/(1-\alpha)}|\mathcal{Y}|$, resp.
Additionally, here we slightly modify the structure between $\mathcal{Y}$ and $u_{*}$ as follows: in $H_i$ where $0\le i\le p$, $i/p\cdot|\mathcal{Y}|$ nodes in $\mathcal{Y}$ have an outgoing edge to the second level in the structure and other nodes in $\mathcal{Y}$ only have a self-loop, while the edges in other levels are all present.

Since the maximum in-degree in the two structures is $\Theta(1)$, this modification removes the constraint that $|\mathcal{V}|=O(\Deltain)$.
Also, the lower bound for estimating the PageRank score of $t$ remains $\Omega\left(\min\big(d|\mathcal{V}|,n/|\mathcal{Y}|\big)\right)$.
Crucially, we need to guarantee that $\pi_i(t)=\big(1+\Omega(1)\big)\pi_{i-1}(t)$ still holds for each $1\le i\le p$.
By the property of the multi-level structure, we can compute that the PageRank score of $u_{*}$ is $\Theta\left(k^{L_{\mathcal{Y}}}/n\right)$ and $\pi_i(t)=\Theta\left(k^{L_{\mathcal{V}}}/n+(1-\alpha)^{L_{\mathcal{V}}}\cdot k^{L_{\mathcal{Y}}}/n\right)$.
We also have $\pi_i(t)-\pi_{i-1}(t)=\Theta\left((1-\alpha)^{L_{\mathcal{V}}}\cdot k^{L_{\mathcal{Y}}}/n\right)$.
In light of these equations, we set $k^{L_{\mathcal{V}}}=(1-\alpha)^{L_{\mathcal{V}}}\cdot k^{L_{\mathcal{Y}}}$ to ensure that $\pi_i(t)=\big(1+\Omega(1)\big)\pi_{i-1}(t)$ for each $1\le i\le p$.
This setting is equivalent to
\begin{align*}
    L_{\mathcal{Y}}=\log_{k}\left(\frac{k}{1-\alpha}\right)^{L_{\mathcal{V}}}=\log_{k}|\mathcal{V}|,
\end{align*}
and is further equivalent to
\begin{align*}
    \quad|\mathcal{Y}|=\left(\frac{k}{1-\alpha}\right)^{L_{\mathcal{Y}}}=\left(\frac{k}{1-\alpha}\right)^{\log_{k}|\mathcal{V}|}=|\mathcal{V}|^{1-\log_{k}(1-\alpha)}.
\end{align*}
We define $\beta=\log_{k}(1-\alpha)$ for simplicity.
To maximize $\min\big(d|\mathcal{V}|,n/|\mathcal{Y}|\big)$, we force $d|\mathcal{V}|=\Theta\big(n/|\mathcal{Y}|\big)$, which leads to
\begin{align} \label{eqn:setting_V}
    |\mathcal{V}|=\Theta\left(\left(\frac{n}{d}\right)^{1/(2-\beta)}\right).
\end{align}
Thus, for a specified $d$, this setting gives the lower bound of
\begin{align*}
    \Omega\left(n^{1/(2-\beta)}d^{1-1/(2-\beta)}\right).
\end{align*}
We conclude our setting and lower bound in this case as follows.

\header{\textbf{Case 2:} $\min(\Deltain,\Deltaout)=o\left(n^{1/3}\right)$.}

We set $d=\Theta\big(\min(\Deltain,\Deltaout)\big)$ and set $|\mathcal{V}|$ according to Equation~\eqref{eqn:setting_V}.
The resultant lower bound is $\Omega\left(n^{1/(2-\beta)}\big(\min(\Deltain,\Deltaout)\big)^{1-1/(2-\beta)}\right)$.
Recalling that $\beta=\log_{k}(1-\alpha)$ and $k$ is a positive integer, 
we obtain the claimed lower bound of $\Omega\left(n^{1/2-\smallexpo}\big(\min(\Deltain,\Deltaout)\big)^{1/2+\smallexpo}\right)$ as long as we set $k$ to be sufficiently large.

This finishes the whole proof.
\end{proof}

\header{\textbf{Remark.}}
We clarify that our Theorem~\ref{thm:SNPR_lower_bound_Delta_in} addresses the asymptotic behavior of $\Deltain$ and does not apply when $\Deltain$ is limited by a fixed constant such as $2$ or $3$.
In such cases, the computational complexity of single-node PageRank estimation strongly depends on the value of $\alpha$ and can beat the $\Theta\big(\sqrt{n}\big)$ barrier for large $\alpha$.
For example, when $\Deltain$ is a fixed constant and $\alpha$ is close to $1$, the probability mass to be reversely pushed from $t$ in \bpush decays rapidly.
Thus, the \bpush process terminates early and the algorithm only needs to explore a small number of $t$'s ancestors to estimate $\pi(t)$.
However, this phenomenon is a consequence of the geometric damping property of PageRank instead of the intrinsic property of random walks.
If we define the centrality of $t$ as $1/(nL)\cdot\sum_{v\in V}\sum_{\ell=0}^{L}P^{(\ell)}(v,t)$, where $L=\lceil\log_2 n\rceil$ and $P^{(\ell)}(v,t)$ is the $\ell$-hop transition probability from $v$ to $t$, then our arguments still yield a matching lower bound of $\Omega\big(\sqrt{n}\big)$.
In fact, the scenario of fixed constant $\Deltain$ has been considered by \cite{bar2008local} before, giving a lower bound of $\Omega\big(n^{1/2-\smallexpo}\big)$ under the node-centric graph-access model when $\Deltain=2$, where $\smallexpo>0$ is a small constant that depends on $\alpha$.
More careful analyses under such settings are beyond the scope of this paper.

%% file: acknowledgments.tex
\section*{Acknowledgments}

This research was supported by National Natural Science Foundation of China (No. U2241212, No. 61932001, No. U2001212), Beijing Natural Science Foundation (No. 4222028), Beijing Outstanding Young Scientist Program No.BJJWZYJH012019100020098, Alibaba Group through Alibaba Innovative Research Program, and Huawei-Renmin University joint program on Information Retrieval.
We also acknowledge the support provided by the fund for building world-class universities (disciplines) of Renmin University of China and by the funds from Engineering Research Center of Next-Generation Intelligent Search and Recommendation, Ministry of Education, from Intelligent Social Governance Interdisciplinary Platform, Major Innovation \& Planning Interdisciplinary Platform for the ``Double-First Class'' Initiative, Renmin University of China, from Public Policy and Decision-making Research Lab of Renmin University of China, and from Public Computing Cloud, Renmin University of China.
Finally, we extend our gratitude to the anonymous reviewers for their valuable comments.

%% file: appendix.tex
\section{Appendix}

\subsection{Table of Notations and Pseudocodes} \label{sec:table_notations}

Table~\ref{tbl:def-notation} summarizes the frequently used notations in this paper, and Algorithms~\ref{alg:SampleNode} gives the pseudocodes for \SampleNode.

\begin{table*}[ht]
    \centering
    \renewcommand{\arraystretch}{1.3}
    \caption{Table of notations.} \label{tbl:def-notation}
    \begin{tabular}{ll}
    \toprule
    \textbf{Notation} & \textbf{Description} \\
    \midrule
    $G=(V,E)$ & underlying directed graph with node set $V$ and edge set $E$ (Section~\ref{sec:preliminaries}) \\
    $n, m$ & number of nodes and edges in $G$ (Section~\ref{sec:preliminaries}) \\
    $\Nin(v),\Nout(v)$ & set of parents and children of $v$ (Section~\ref{sec:preliminaries}) \\
    $\din(v),\dout(v)$ & in-degree and out-degree of $v$ (Section~\ref{sec:preliminaries}) \\
    $\Deltain,\Deltaout$ & maximum in-degree and out-degree of $G$ (Section~\ref{sec:preliminaries}) \\
    $\alpha$ & decay factor in defining PageRank and PPR, $\alpha\in(0,1)$ (Section~\ref{sec:preliminaries_PageRank}) \\
    $\pi(v),\pi(u,v)$ & PageRank score of $v$, PPR score of $v$ w.r.t. $u$ (Section~\ref{sec:preliminaries_PageRank}) \\
    $\delta$ & parameter in defining the $\delta$-contributing set (Section~\ref{sec:preliminaries_contribution}) \\
    \midrule
    $\rmax$ & threshold parameter in \bpush (Section~\ref{sec:preliminaries_bpush}) \\
    $\epib{},\rb{}$ & reserves and residues in \bpush (Section~\ref{sec:preliminaries_bpush}) \\
    \midrule
    $c$ & relative error parameter (Section~\ref{sec:SNPR_upper_bound}) \\
    $p_f$ & failure probability parameter (Section~\ref{sec:SNPR_upper_bound}) \\
    $\chi_v$ & indicator variable of $\SampleNode()=v$ (Section~\ref{sec:SNPR_upper_bound}) \\
    $q(t)$ & bidirectional estimator (Section~\ref{sec:SNPR_upper_bound}) \\
    $n_r$ & number of random walk simulations in \BiPPR (Section~\ref{sec:SNPR_upper_bound}) \\
    $\epi(t)$ & estimate of $\pi(t)$ returned by \BiPPR (Section~\ref{sec:SNPR_upper_bound}) \\
    \bottomrule
    \end{tabular}
\end{table*}

\begin{algorithm}[ht]
\DontPrintSemicolon
\caption{$\SampleNode()$} \label{alg:SampleNode}
    \KwOut{a sampled node $v$, where each node $v\in V$ is sampled w.p. $\pi(v)$}
    $v\gets\jump()$ \;
    \While{\textup{True}}
    {
        with probability $\alpha$ \Return $v$ \;
        $v\gets\child\Big(v,\randint\big(\outdeg(v)\big)\Big)$ \;
        \textcolor{gray}{// $\randint\big(\outdeg(v)\big)$ returns a uniformly random integer in $\big\{1,2,\dots,\outdeg(v)\big\}$} \;
    }
\end{algorithm}

\subsection{Deferred Proofs} \label{sec:deferred_proofs}

\begin{proof}[Proof of Theorem~\ref{thm:previous_cost_bp}] \hypertarget{proof:previous_cost_bp}
Recall that when a pushback operation is performed on $v$, $v$ sends probability mass to each of its parents, which takes $O\big(\din(v)\big)$ time.
We define $\SP(v)$ (shorthand for ``send probability'') to be the number of pushback operations performed on $v$ and $\RP(v)$ (shorthand for ``receive probability'') to be the number of times $v$ receives probability mass from its children.
It holds that $\RP(v)=\sum_{u\in\Nout(v)}\SP(u)$.
By the procedure of \bpush (Algorithm~\ref{alg:BP}), its complexity can be bounded by both
\begin{align*}
    O\left(\sum_{v\in V}\SP(v)\cdot\din(v)\right)\quad\text{and}\quad O\left(\sum_{v\in V}\RP(v)\right).
\end{align*}
Now recall from Theorem~\ref{thm:properties_bp} that the reserve $\epib{v}$ is always an underestimate of $\pi(v,t)$.
By the procedure of \bpush, each time a pushback operation is performed on $v$, its reserve is increased by at least $\alpha\rmax$.
Thus, $\SP(v)\le\pi(v,t)/(\alpha\rmax)=O\big(\pi(v,t)/\rmax\big)$.
Using $\RP(v)=\sum_{u\in\Nout(v)}\SP(u)$ and Inequality~\eqref{ineqn:sum_PPR_children}, we further have
\begin{align*}
    \RP(v)=O\left(\sum_{u\in\Nout(v)}\frac{\pi(u,t)}{\rmax}\right)=O\left(\frac{\pi(v,t)\cdot\dout(v)}{\rmax}\right).
\end{align*}
Substituting into $O\left(\sum_{v\in V}\SP(v)\cdot\din(v)\right)$ and $O\left(\sum_{v\in V}\RP(v)\right)$, we respectively obtain the complexity bound of 
\begin{align*}
    O\left(\frac{1}{\rmax}\sum_{v\in V}\pi(v,t)\cdot\din(v)\right)\quad\text{and}\quad O\left(\frac{1}{\rmax}\sum_{v\in V}\pi(v,t)\cdot\dout(v)\right),
\end{align*}
finishing the proof.
\end{proof}

\begin{proof}[Proof of Theorem~\ref{thm:contribution_upper_bound}] \hypertarget{proof:contribution_upper_bound}
First, note that one can trivially use $O(n+m)=O(m)$ time to explore the whole graph and return $V$ as a result.
Therefore, it suffices to prove the upper bound of $O\left(\min\big(\Deltain,\Deltaout,\sqrt{m}\big)/\delta\right)$.

For ease of presentation, in the following, we give three algorithms to separately achieve the upper bounds of $O(\Deltain/\delta)$, $O(\Deltaout/\delta)$, and $O\big(\sqrt{m}/\delta\big)$.
One can straightforwardly combine them into one algorithm that runs in $O\left(\min\big(\Deltain,\Deltaout,\sqrt{m}\big)/\delta\right)$ time.
The three algorithms share a unified framework: we repeatedly run \bpush with $\rmax=1,1/2,1/4,1/8,...$ and record a quantity $T_{\rmax}$ for each running ($T_{\rmax}$ depends on the running process of $\bpush(t,\alpha,\rmax)$ and is to be determined shortly), and once the sum of $T_{\rmax}$ is about to surpass $\Theta(1/\delta)$, we terminate the \bpush process, restore the results of the last running of it, and return all nodes with nonzero reserves as the output.
Our settings of $T_{\rmax}$ will ensure that recording it does not increase the running cost of \bpush asymptotically.
To establish our results, it suffices to guarantee that:
\begin{enumerate}
    \item The quantity $T_{\rmax}$ can be upper bounded by $O\big(n\pi(t)/\rmax\big)$;
    \item The cost of running $\bpush(t,\alpha,\rmax)$ is bounded by $O(\Deltain\cdot T_{\rmax})$, $O(\Deltaout\cdot T_{\rmax})$, or $O(\sqrt{m}\cdot T_{\rmax})$, correspondingly.
\end{enumerate}
To explain, the first condition ensures that the sum of $T_{\rmax}$ for running \bpush with $\rmax=1,1/2,1/4,\dots,\rmax_{*}$ is bounded by $O(1/\delta)$, where $\rmax_{*}$ is the largest value of the form $1\big/2^k$ such that $\rmax_{*}<\delta n\pi(t)$, for positive integers $k$.
Thus, the final value of $\rmax$ obtained by the algorithm is smaller than $\delta n\pi(t)$, ensuring that the output contains the $\delta$-contributing set of $t$.
On the other hand, the second condition ensures that the cost of the algorithm is bounded by $O(\Deltain/\delta)$, $O(\Deltaout/\delta)$, and $O\big(\sqrt{m}/\delta\big)$, correspondingly, since the sum of $T_{\rmax}$ in running the algorithm is $\Theta(1/\delta)$.
It remains to set $T_{\rmax}$ and check that it satisfies the two conditions.
To this end, we will use the notions and properties of $\SP(v)$ and $\RP(v)$ in running $\bpush(t,\alpha,\rmax)$, as described in the \hyperlink{proof:previous_cost_bp}{proof} of Theorem~\ref{thm:previous_cost_bp}.

To achieve the bound of $O(\Deltain/\delta)$, we set $T_{\rmax}=\sum_{v\in V}\SP(v)$.
For condition 1, recall that $\SP(v)=O\big(\pi(v,t)/\rmax\big)$, yielding $T_{\rmax}=O\left(\sum_{v\in V}\pi(v,t)/\rmax\right)=O\big(n\pi(t)/\rmax\big)$.
For condition 2, recall that the cost of \bpush is bounded by $O\left(\sum_{v\in V}\SP(v)\cdot\din(v)\right)$, which is further bounded by $O\left(\sum_{v\in V}\SP(v)\cdot\Deltain\right)=O(\Deltain\cdot T_{\rmax})$.

To achieve the bound of $O(\Deltaout/\delta)$, we set $T_{\rmax}=\sum_{v\in V}\RP(v)/\dout(v)$.
For condition 1, recall that $\RP(v)=O\big(\pi(v,t)\cdot\dout(v)/\rmax\big)$, yielding $T_{\rmax}=O\left(\sum_{v\in V}\pi(v,t)/\rmax\right)=O\big(n\pi(t)/\rmax\big)$.
For condition 2, recall that the cost of \bpush is bounded by $O\left(\sum_{v\in V}\RP(v)\right)$, which is further bounded by $O(\Deltaout\cdot T_{\rmax})$.

To achieve the bound of $O\big(\sqrt{m}/\delta\big)$, we define
\begin{align*}
    f(v)=\sum_{u\in\Nin(v)}\dout(u),\quad g(v)=\sum_{u\in\Nin(v)}\frac{1}{\dout(u)},
\end{align*}
and we set
\begin{align*}
    T_{\rmax}=\left(\sum_{v\in V}\SP(v)\right)^{1/2}\left(\sum_{v\in V}\SP(v)\cdot g(v)\right)^{1/2}.
\end{align*}
For condition 1, we have
\begin{align*}
    T_{\rmax}&=O\left(\left(\sum_{v\in V}\frac{\pi(v,t)}{\rmax}\right)^{1/2}\left(\sum_{v\in V}\frac{\pi(v,t)}{\rmax}\cdot g(v)\right)^{1/2}\right) \\
    &=O\left(\left(\frac{n\pi(t)}{\rmax}\right)^{1/2}\left(\frac{1}{\rmax}\sum_{u\in V}\frac{1}{\dout(u)}\sum_{v\in\Nout(u)}\pi(v,t)\right)^{1/2}\right) \\
    &=O\left(\left(\frac{n\pi(t)}{\rmax}\right)^{1/2}\left(\sum_{u\in V}\frac{\pi(u,t)}{\rmax}\right)^{1/2}\right)=O\left(\frac{n\pi(t)}{\rmax}\right),
\end{align*}
where we used Inequality~\eqref{ineqn:sum_PPR_children}.
For condition 2, 
using a similar derivation as in the \hyperlink{proof:cost_bp}{proof} of Theorem~\ref{thm:cost_bp}, we can derive that the cost of \bpush is bounded by
\begin{align*}
    O\left(\sum_{v\in V}\SP(v)\cdot\din(v)\right)=O\left(\left(\sum_{v\in V}\SP(v)\cdot f(v)\right)^{1/2}\left(\sum_{v\in V}\SP(v)\cdot g(v)\right)^{1/2}\right).
\end{align*}
As $f(v)\le m$, this can be further bounded by 
\begin{align*}
    O\left(\sqrt{m}\left(\sum_{v\in V}\SP(v)\right)^{1/2}\left(\sum_{v\in V}\SP(v)\cdot g(v)\right)^{1/2}\right)=O\big(\sqrt{m}\cdot T_{\rmax}\big),
\end{align*}
finishing the proof.
\end{proof}